\documentclass[a4paper,UKenglish,cleveref, autoref, thm-restate]{lipics-v2021}

\usepackage{cite}
\usepackage{hyperref}
\hypersetup{
    colorlinks=true,
    linkcolor=red,
    filecolor=magenta,
    citecolor = blue,
    urlcolor=cyan,
    linktocpage,
    }

\title{Algorithms for Computing Closest Points for Segments} 

\author{Haitao Wang}{Kahlert School of Computing,
University of Utah, Salt Lake City, UT 84112, USA \and \url{https://www.cs.utah.edu/~hwang} }{haitao.wang@utah.edu}{https://orcid.org/0000-0001-8134-7409}{}

\authorrunning{H. Wang} 

\Copyright{Haitao Wang} 



\category{} 

\relatedversion{} 

\funding{This research was supported in part by NSF under Grant CCF-2300356.}

\nolinenumbers 

\hideLIPIcs
\EventEditors{John Q. Open and Joan R. Access}
\EventNoEds{2}
\EventLongTitle{42nd Conference on Very Important Topics (CVIT 2016)}
\EventShortTitle{CVIT 2016}
\EventAcronym{CVIT}
\EventYear{2016}
\EventDate{December 24--27, 2016}
\EventLocation{Little Whinging, United Kingdom}
\EventLogo{}
\SeriesVolume{42}
\ArticleNo{23}

\newcommand{\poly}{\text{poly}}
\def\scrD{\mathscr{D}}
\def\calA{\mathcal{A}}
\newcommand{\vd}{\mathsf{VD}}
\newcommand{\ch}{\mathsf{CH}}

\begin{document}

\maketitle

\begin{abstract}
Given a set $P$ of $n$ points and a set $S$ of $n$ segments in the plane, we
consider the problem of computing for each segment of $S$ its closest point in
$P$. The previously best algorithm solves the problem in
$n^{4/3}2^{O(\log^*n)}$ time [Bespamyatnikh, 2003] and a lower bound (under a somewhat restricted model)
$\Omega(n^{4/3})$ has also been proved.
In this paper, we present an $O(n^{4/3})$ time algorithm and
thus solve the problem optimally (under the restricted model).
In addition, we also present data structures for solving the online version of the
problem, i.e., given a query segment (or a line as a special case), find its closest point in $P$.
Our new results improve the previous work.
\end{abstract}

\section{Introduction}
\label{sec:intro}

Given a set $P$ of $n$ points and a set $S$ of $n$ segments in the plane, we
consider the problem of computing for each segment of $S$ its closest point in
$P$. We call it the {\em segment-closest-point} problem.
Previously, Bespamyatnikh~\cite{ref:BespamyatnikhCo03} gave an $n^{4/3}2^{O(\log^*n)}$ time algorithm for the problem, improving upon an $O(n^{4/3}\log^{O(1)}n)$ time result of Agarwal and Procopiuc.
The problem can be viewed as a generalization of Hopcroft's problem~\cite{ref:AgarwalPa902,ref:ChanHo22,ref:ChazelleCu93,ref:EdelsbrunnerTh90,ref:MatousekRa93}, which is to determine whether any point of a given set of $n$ points lies on any of the given $n$ lines.
Erickson~\cite{ref:EricksonNe96} proved an $\Omega(n^{4/3})$ time lower bound
for Hopcroft's problem under a somewhat restricted {\em partition model}. This implies the
	same lower bound on the segment-closest-point problem.
For Hopcroft's problem, Chan and Zheng~\cite{ref:ChanHo22} recently gave an
$O(n^{4/3})$ time algorithm, which matches the lower bound and thus is optimal.

In this paper, with some new observations on the problem as well as the
techniques from Chan and Zheng~\cite{ref:ChanHo22} (more specifically, the {\em
$\Gamma$-algorithm framework} for bounding algebraic decision tree complexities), we
present a new algorithm that solves the segment-closest-point problem in $O(n^{4/3})$ time and thus
is optimal under Erickson's partition model~\cite{ref:EricksonNe96}.
It should be noted that our result is not a direct application of Chan and
Zheng's techniques~\cite{ref:ChanHo22}, but rather many new observations and
techniques are
needed. For example, one subroutine in our problem is the following
{\em outside-hull segment queries}: Given a segment
outside the convex hull of $P$, find its
closest point in $P$. Bespamyatnikh and Snoeyink~\cite{ref:BespamyatnikhQu00}
built a data structure in $O(n)$ space and $O(n\log n)$ time such that each
query can be answered in $O(\log n)$ time. Unfortunately, their query algorithm
does not fit the $\Gamma$-algorithm framework of Chan and
Zheng~\cite{ref:ChanHo22}. To resolve the issue, we develop another
algorithm for the problem based on new observations.
Our approach is simpler, and more importantly, it
fits the $\Gamma$-algorithm framework of Chan and Zheng~\cite{ref:ChanHo22}. The
result may be interesting in its own right.

We believe the $\Gamma$-algorithm framework have applications for a lot of problems in computational geometry. Very recently, Chan, Cheng, and Zheng~\cite{ref:ChanAn24} used the framework to tackle the higher-order Voronoi diagram problem.
In addition to \cite{ref:ChanAn24,ref:ChanHo22}, our result is another demonstration of this technique and it helps to build our understanding along this line of work.

We also consider the online version of the problem, called {\em the segment
query problem}: Preprocess $P$ so that given a query segment,
its closest point in $P$ can be found efficiently.
For the special case where the query segment is outside the convex hull of $P$, one can use the data structure of Bespamyatnikh and Snoeyink~\cite{ref:BespamyatnikhQu00} mentioned above.
For simplicity, we use $(T_1(n),T_2(n),T_3(n))$ to denote the complexity of a
data structure if its preprocessing time, space, and query time are on the order
of $T_1(n)$, $T_2(n)$, and $T_3(n)$, respectively. Using this notation, the
complexity of the above data structure of Bespamyatnikh and
Snoeyink~\cite{ref:BespamyatnikhQu00} is $O(n\log n, n, \log n)$.
The general problem, however, is much more challenging.
Goswami, Das, and Nandy~\cite{ref:GoswamiTr04}'s method yields a result of
complexity $O(n^2,n^2,\log^2 n)$.
We present a new data structure of complexity $O(nm(n/m)^{\delta},nm\log (n/m), \sqrt{n/m}\log (n/m))$,
any $m$ with $1\leq m\leq n\log^2\log n/\log^4 n$ and any $\delta>0$.
Note that for the large space case (i.e., when $m=n\log^2\log n/\log^4 n$), the complexity of our data structure is $O(n^2/\log^{4-\delta}n,n^2\log^3\log n/\log^4 n,\log^2 n)$, which improves the above result
of~\cite{ref:GoswamiTr04} on the preprocessing time and space by a factor of roughly $\log^4 n$.
We also present a faster randomized data structure of complexity $O(nm\log (n/m),nm\log (n/m),
\sqrt{n/m})$ for any $m$ with $1\leq m\leq n/\log^4 n$, where the
preprocessing time is expected and the query time holds with high probability.
In addition, using Chan's randomized techniques~\cite{ref:ChanGe99} and Chan and Zheng's recent randomized result on triangle range counting~\cite{ref:ChanHo22}, we can obtain a randomized data structure of complexity $O(n^{4/3},n^{4/3},n^{1/3})$.\footnote{The idea was suggested by an anonymous reviewer.} Note that this data structure immediately leads to a randomized algorithm of $O(n^{4/3})$ expected time for the segment-closest-point problem. As such, for solving the segment-closest-point problem, our main effort is to derive an $O(n^{4/3})$ deterministic time algorithm. Note that this is aligned with the motivation of proposing the $\Gamma$-algorithm framework in~\cite{ref:ChanHo22}, whose goal was to obtain an $O(n^{4/3})$ deterministic time algorithm for Hopcroft's problem although a much simpler randomized algorithm of $O(n^{4/3})$ expected time was already presented.

If each query segment is a line, we call it the {\em line query problem}, which
has been extensively studied.
Previous work includes Cole and Yap~\cite{ref:ColeGe83}'s and
Lee and Ching~\cite{ref:LeeTh85}'s data structures of
complexity $O(n^2,n^2,\log n)$, Mitra and Chaudhuri~\cite{ref:MitraEf98}'s work of
complexity $O(n\log n, n, n^{0.695})$, Mukhopadhyay~\cite{ref:MukhopadhyayUs03}'s result of complexity
$O(n^{1+\delta},n\log n, n^{1/2+\delta})$ for any $\delta>0$.
As observed by Lee and Ching~\cite{ref:LeeTh85}, the problem can be reduced to
vertical ray-shooting in the dual plane, i.e., finding the first
line hit by a query vertical ray among a given set of $n$ lines (see
Section~\ref{sec:linequery} for the details). Using
the ray-shooting algorithms, the best deterministic result is
$O(n^{1.5},n,\sqrt{n}\log n)$~\cite{ref:WangAl20} while the best randomized result is $O(n\log n, n, \sqrt{n})$~\cite{ref:ChanSi23};
refer to ~\cite{ref:AgarwalAp93,ref:Bar-YehudaVa94,ref:ChanHo22,ref:ChengAl92} for other
(less efficient) work on ray-shootings. We build a new deterministic data
structure of complexity $O(nm(n/m)^{\delta},nm\log (n/m), \sqrt{n/m})$,
for any $1\leq m\leq n/\log^2 n$.
We also have another faster randomized result of complexity
$O(nm\log (n/m),nm\log (n/m), \sqrt{n/m})$, for any $m$ with $1\leq m\leq
n/\log^2 n$, where the preprocessing time is expected while the query time
holds with high probability. Our
results improve all previous work except the randomized result of Chan and
Zheng~\cite{ref:ChanSi23}.
For example, if $m=1$, our data structure is the only deterministic
one whose query time is $O(\sqrt{n})$ with near linear space; if $m=n/\log^2 n$,
our result achieves $O(\log n)$ query time while the preprocessing time and
space are all subquadratic, better than those by Cole and
Yap~\cite{ref:ColeGe83} and Lee and Ching~\cite{ref:LeeTh85}.

\subparagraph{Other related work.}
If all segments are pairwise disjoint, then the segment-closest-point
problem was solved in $O(n\log^2 n)$ time by
Bespamyatnikh~\cite{ref:BespamyatnikhCo03}, improving over the $O(n\log^3 n)$ time
algorithm of Bespamyatnikh and Snoeyink~\cite{ref:BespamyatnikhQu00}.

If every segment of $S$ is a single point, then the problem can be easily solved in
$O(n\log n)$ time using the Voronoi diagram of $P$.
Also, for any segment $s\in S$, if the point of $s$ closest to $P$ is an endpoint of $s$,
then finding the closest point of $s$ in $P$ can be done using the Voronoi
diagram of $P$. Hence, the remaining issue is to find the first point of $P$
hit by $s$ if we drag $s$ along the directions perpendicularly to $s$. If all
segments of $S$ have the same slope, then the problem can be solved in
$O(n\log n)$ time using the segment dragging query data structure of
Chazelle~\cite{ref:ChazelleAn88}, which can answer each query in $O(\log n)$
time after $O(n)$ space and $O(n\log n)$ time preprocessing. However, the
algorithm~\cite{ref:ChazelleAn88} does not work if the query segments have
arbitrary slopes. As such, the challenge of the problem is to solve the dragging
queries for all segments of $S$ when their slopes are not the same.

The {\em segment-farthest-point} problem has also been studied, where
one wants to find for each segment of $S$ its farthest point in $P$. The problem
appears much easier. For the line query problem (i.e., given a query line, find its farthest point in $P$),
Daescu et al.~\cite{ref:DaescuFa06} gave a data structure of
complexity $O(n\log n, n, \log n)$. Using this result, they also
proposed a data structure of complexity $O(n\log n, n\log n, \log^2
n)$ for the segment query problem.
Using this segment query data structure, the segment-farthest-point can
be solved in $O(n\log^2 n)$ time.

\subparagraph{Outline.}
The rest of the paper is organized as follows. In Section~\ref{sec:pre}, we
introduce some notation and concepts. In Section~\ref{sec:segment}, we
present our $O(n^{4/3})$ deterministic time algorithm for the segment-closest-point problem.
We actually solve a more general problem where the number of points is not equal
to the number of segments, referred to as the {\em asymmetric case}, and our
algorithm runs in $O(n^{2/3}m^{2/3}+n\log n+m\log^2 n)$ time with $n$ as the
number of points and $m$ as the number of segments.
We present a simpler algorithm for the
line case of the problem in Section~\ref{sec:line} where all segments are lines,
and the algorithm also runs
in $O(n^{4/3})$ time (and $O(n^{2/3}m^{2/3}+(n+m)\log n)$ time for the
asymmetric case). The line query problem is discussed in
Section~\ref{sec:linequery} while the segment query problem is solved in
Section~\ref{sec:segmentquery}.

\section{Preliminaries}
\label{sec:pre}

For two closed subsets $A$ and $B$ in the plane, let $d(A,B)$ denote the minimum
distance between any point of $A$ and any point of $B$. The point $p$ of $A$ closest
to $B$, i.e., $d(p,B)=d(A,B)$, is called the {\em closest point} of $B$ in $A$.


For any two points $a$ and $b$ in the plane, we use $\overline{ab}$ to denote
the segment with $a$ and $b$ as its two endpoints.

For any point $p$ in the plane, we use $x(p)$ and $y(p)$ to denote its $x$- and
$y$-coordinates, respectively. For a point $p$ and a region $A$ in the plane, we
say that $p$ is {\em to the left} of $A$ if $x(p)\leq x(q)$ for all points $q\in
A$, and $p$ is {\em strictly to the left} of $A$ if $x(p)< x(q)$ for all points $q\in
A$; the concepts {\em (strictly) to the right} is defined symmetrically.

For a set $Q$ of points in the plane, we usually use $\vd(Q)$ to denote the Voronoi diagram of $Q$ and use $\ch(Q)$ to denote the convex hull of $Q$; we also use $Q(A)$ to denote the subset of $Q$ in $A$, i.e., $Q(A)=Q\cap A$, for any region $A$ in the plane.

\subparagraph{Cuttings.}
Let $H$ be a set of $n$ lines in the plane. Let $H_A$ denote the subset of
lines of $H$ that intersect the interior of $A$ (we also say that these lines
{\em cross} $A$), for a compact region $A$ in the plane.
A {\em cutting} is a collection $\Xi$ of closed cells (each of which is a triangle) with disjoint interiors, which together cover the entire plane~\cite{ref:ChazelleCu93,ref:MatousekRa93}.
The {\em size} of $\Xi$ is the number of cells in $\Xi$. For a parameter $r$ with $1\leq r\leq n$, a {\em $(1/r)$-cutting} for $H$ is a cutting $\Xi$ satisfying $|H_{\sigma}|\leq n/r$ for every cell $\sigma\in \Xi$.



A cutting $\Xi'$ {\em $c$-refines} another cutting $\Xi$ if every cell of $\Xi'$ is contained in a
single cell of $\Xi$ and every cell of $\Xi$ contains at most $c$
cells of $\Xi'$. A {\em hierarchical $(1/r)$-cutting} for $H$ (with two
constants $c$ and $\rho$) is a sequence
of cuttings $\Xi_0,\Xi_1,\ldots,\Xi_k$ with the following properties.
$\Xi_0$ is the entire plane. For each $1\leq i\leq k$, $\Xi_i$ is a
$(1/\rho^i)$-cutting of size $O(\rho^{2i})$ which $c$-refines
$\Xi_{i-1}$. In order to make
$\Xi_k$ a $(1/r)$-cutting, we set $k=\lceil\log_{\rho} r\rceil$. Hence, the size of the last cutting
$\Xi_k$ is $O(r^2)$. If a cell $\sigma\in
\Xi_{i-1}$ contains a cell $\sigma'\in \Xi_i$, we say that $\sigma$ is
the {\em parent} of $\sigma'$ and $\sigma'$ is a {\em child} of
$\sigma$. As such, one could view $\Xi$ as a tree in which each node corresponds
to a cell $\sigma\in \Xi_i$, $0\leq i\leq k$.

For any $1\leq r\leq n$, a hierarchical $(1/r)$-cutting of size
$O(r^2)$ for $H$ (together with $H_{\sigma}$ for every cell
$\sigma$ of $\Xi_i$ for all $i=0,1,\ldots,k$) can be computed in
$O(nr)$ time~\cite{ref:ChazelleCu93}. Also, it is easy to check that $\sum_{i=0}^k\sum_{\sigma\in \Xi_i}|H_{\sigma}|=O(nr)$.

\section{The segment-closest-point problem}
\label{sec:segment}

In this section, we consider the segment-closest-point problem. Let $P$ be a set
of $n$ points and $S$ a set of $n$ segments in the plane. The problem is to
compute for each segment of $S$ its closest point in $P$. We
make a general position assumption that no segment of $S$ is vertical (for a vertical segment, its closest point can be easily found, e.g., by building a segment dragging query data structure~\cite{ref:ChazelleAn88} along with the Voronoi diagram of $P$).

%

We start with a review of an algorithm of
Bespamyatnikh~\cite{ref:BespamyatnikhCo03}, which will be needed in our new
approach.

\subsection{A review of Bespamyatnikh's algorithm~\cite{ref:BespamyatnikhCo03}}
\label{sec:review}

As we will deal with subproblems in which the number of lines is not equal to
the number of segments, we let $m$ denote the number of segments in $S$ and $n$
the number of points in $P$. As such, the size of our original problem $(S,P)$ is $(m,n)$.

Let $H$ be the set of the supporting lines of the segments of $S$.
For a parameter $r$ with $1\leq r\leq \min\{m,\sqrt{n}\}$, compute a hierarchical
$(1/r)$-cutting $\Xi_0,\Xi_1,\ldots,\Xi_k$ for $H$.
For each cell $\sigma \in \Xi_i$, $0\leq i\leq k$, let $P(\sigma)=P\cap \sigma$,
i.e., the subset of the points of $P$ in $\sigma$; let $S(\sigma)$ denote the
subset of the segments of $S$ intersecting $\sigma$.
We further partition each cell of $\Xi_k$ into triangles so that each triangle
contains at most $n/r^2$ points of $P$ and the number of new triangles in
$\Xi_k$ is still bounded by $O(r^2)$. For convenience, we consider the new
triangles as new cells of $\Xi_k$ (we still define $P(\sigma)$ and $S(\sigma)$
for each new cell $\sigma$ in the same way as above; so now $|P(\sigma)|\leq
	n/r^2$ and $|S(\sigma)|\leq m/r$ hold for each cell $\sigma\in \Xi_k$).

For each cell $\sigma\in \Xi_k$, form a subproblem $(S(\sigma),P(\sigma))$ of
size $(m/r,n/r^2)$, i.e., find for each segment $s$ of $S(\sigma)$ its closest
point in $P(\sigma)$. After the subproblem is solved, to find the closest point
of $s$ in $P$, it suffices to find its closest point in $P\setminus P(\sigma)$.
To this end, observe that $P\setminus P(\sigma)$ is exactly the union of
$P(\sigma'')$ for all cells $\sigma''$ such that $\sigma''$ is a child of an
ancestor $\sigma'$ of $\sigma$ and $s\not\in S(\sigma'')$.
As such,
for each of such cells $\sigma''$, find the closest point of $s$
in $P(\sigma'')$. For this, since $s\not\in S(\sigma'')$, $s$ is outside
$\sigma''$ and thus is outside the convex hull of $P(\sigma'')$. Hence, finding
the closest point of $s$ in $P(\sigma'')$ is an outside-hull segment query and
thus the data structure of Bespamyatnikh and Snoeyink~\cite{ref:BespamyatnikhQu00}
(referred to as {\em the BS data structure} in the rest of the paper) is used,
which takes $O(|P(\sigma'')|)$ space and $O(|P(\sigma'')|\log |P(\sigma'')|)$
time preprocessing and can answer each query in $O(\log |P(\sigma'')|)$ time.
More precisely, the processing can be done in $O(|P(\sigma'')|)$ time if the
Voronoi diagram of $P(\sigma'')$ is known.

For the time analysis, let $T(m,n)$ denote the time of the
algorithm for solving a problem of size $(m,n)$.
Then, solving all subproblems takes $O(r^2)\cdot
T(m/r,n/r^2)$ time as there are $O(r^2)$ subproblems of size $(m/r,n/r^2)$.
Constructing the hierarchical cutting as well as computing $S(\sigma)$
for all cells $\sigma$ in all cuttings $\Xi_i$, $0\leq i\leq k$, takes
$O(mr)$ time~\cite{ref:ChazelleCu93}. Computing $P(\sigma)$ for
all cells $\sigma$ can be done in $O(n\log r)$ time.
Preprocessing for constructing the BS data structure for $P(\sigma)$ for all cells $\sigma$
can be done in $O(n\log n\log r)$ time as $\sum_{\sigma\in
\Xi_i}|P(\sigma)|=n$ for each $0\leq i\leq k$, and $k=O(\log r)$. We can further reduce the time
to $O(n(\log r+\log n))$ as follows. We build the BS data structure for
cells of the cuttings in a bottom-up manner, i.e., processing cells of
$\Xi_k$ first and then $\Xi_{k-1}$ and so on. After the preprocessing
for $P(\sigma)$ for a cell $\sigma\in \Xi_k$, which takes $O(|P(\sigma)|\log
(n/r^2))$ time since $|P(\sigma)|\leq n/r^2$, the Voronoi diagram of $P(\sigma)$
is available. After the preprocessing for all cells $\sigma$ of $\Xi_k$ is done, for
each cell $\sigma'$ of $\Xi_{k-1}$, to construct the Voronoi diagram
of $P(\sigma')$, merge the Voronoi diagrams of $P(\sigma)$ for
all children $\sigma$ of $\sigma'$. To this end, as $\sigma'$ has $O(1)$ children, the
merge can be done in $O(|P(\sigma')|)$ time by using the algorithm
of Kirkpatrick~\cite{ref:KirkpatrickEf79}, and thus the preprocessing
for $P(\sigma')$ takes only linear time. In this way, the total
preprocessing time for all cells in all cuttings $\Xi_i$, $0\leq i\leq k$, is bounded by
$O(n(\log r+\log (n/r^2)))$ time, i.e., the time spent on cells of
$\Xi_k$ is $O(n\log (n/r^2))$ and the time on other cuttings is
$O(n\log r)$ in total. Note that $\log r+\log (n/r^2)=\log (n/r)$.
As for the outside-hull segment queries,
according to the properties of the hierarchical cutting,
$\sum_{i=0}^{k}\sum_{\sigma\in \Xi_i}|S(\sigma)|=O(mr)$. Hence, the total number of
outside-hull segment queries on the BS data structure is $O(mr)$ and
thus the total query time is $O(mr\log n)$. In summary, the following
recurrence is obtained for any $1\leq r\leq \min\{m,\sqrt{n}\}$:
\begin{align}\label{equ:10}
  T(m,n)= O(n\log ({n}/{r}) + mr\log n)+O(r^2)\cdot T({m}/{r},{n}/{r^2}).
\end{align}

Using the duality, Bespamyatnikh~\cite{ref:BespamyatnikhCo03} gave a second
algorithm (we will not review this algorithm here because it is not relevant to our
new approach) and obtained the following recurrence for any $1\leq r\leq
\min\{n,\sqrt{m}\}$:
\begin{align}\label{equ:20}
  T(m,n)= O(nr\log n + m\log r\log n)+O(r^2)\cdot T({m}/{r^2},{n}/{r}).
\end{align}

By setting $m=n$ and applying \eqref{equ:20} and \eqref{equ:10} in succession (using the same $r$), it follows that
\begin{align*}
  T(n,n)= O(nr\log n)+O(r^4)\cdot T({n}/{r^3},{n}/{r^3}).
\end{align*}
Setting $r=n^{1/3}/\log n$ leads to
\begin{align}\label{equ:30}
  T(n,n)= O(n^{4/3}) + O((n/\log^3 n)^{4/3})\cdot T(\log^3 n,\log^3 n).
\end{align}
The recurrence solves to $T(n,n)=n^{4/3}2^{O(\log^*n)}$, which is the time bounded obtained by  Bespamyatnikh~\cite{ref:BespamyatnikhCo03}.

\subsection{Our new algorithm}

In this section, we improve the algorithm to $O(n^{4/3})$ time.

By applying recurrence \eqref{equ:30} three times we obtain the following:
\begin{align}\label{equ:40}
  T(n,n)= O(n^{4/3}) + O((n/b)^{4/3})\cdot T(b,b),
\end{align}
where $b=(\log\log\log n)^3$.

Using the property that $b$ is tiny, we show in the following that
after $O(n)$ time preprocessing, we can solve each subproblem $T(b,b)$ in
$O(b^{4/3})$ time (for convenience, by slightly abusing the notation, we also use $T(m,n)$ to denote a subproblem of size $(m,n)$). Plugging the result into \eqref{equ:40}, we obtain
$T(n,n)=O(n^{4/3})$.

More precisely, we show that after $O(2^{\poly(b)})$ time preprocessing, where
$\poly(\cdot)$ is a polynomial function, we can solve each $T(b,b)$ using $O(b^{4/3})$ comparisons,
or alternatively, $T(b,b)$ can be solved
by an algebraic decision tree of height $O(b^{4/3})$.  As $b=(\log\log\log
n)^3$, $2^{\poly(b)}$ is bounded by $O(n)$. To turn this into an algorithm under
the standard real-RAM model, we explicitly construct the algebraic decision tree
for the above algorithm (we may also consider this step as part of preprocessing
for solving $T(b,b)$), which can again be done in $O(2^{\poly(b)})$ time.
As such, that after $O(n)$ time preprocessing, we can solve each $T(b,b)$ in $O(b^{4/3})$ time.
In the following, for notational convenience, we will use $n$ to denote $b$, and
our goal is to prove the following lemma.
\begin{lemma}\label{lem:10}
After $O(2^{\poly(n)})$ time preprocessing, $T(n,n)$
can be solved using $O(n^{4/3})$ comparisons.
\end{lemma}

We apply recurrence~\eqref{equ:10} by setting $m=n$ and $r=n^{1/3}$, and obtain the following
\begin{align}\label{equ:50}
  T(n,n)= O(n\log n + n^{4/3}\log n)+O(n^{2/3})\cdot T(n^{2/3},n^{1/3}).
\end{align}
Recall that the term $n^{4/3}\log n$ is due to that there are $O(n^{4/3})$
outside-hull segment queries.
To show that $T(n,n)$ can be solved by $O(n^{4/3})$ comparisons, there are two
challenges: (1) solve all outside-hull segment queries using $O(n^{4/3})$
comparisons; (2) solve each subproblem $T(n^{2/3},n^{1/3})$ using $O(n^{2/3})$
comparisons.

\subparagraph{$\Gamma$-algorithm framework.}
To tackle these challenges, we use a {\em
$\Gamma$-algorithm framework} for bounding decision tree complexities proposed
by Chan and Zheng~\cite{ref:ChanHo22}. We briefly review it here (see
Section~4.1~\cite{ref:ChanHo22} for the details).
Roughly speaking, this framework is an
algorithm that only counts the number of comparisons (called
$\Gamma$-comparisons in~\cite{ref:ChanHo22}) for determining whether a
point belongs to a semialgebraic set of $O(1)$ degree in a constant-dimensional
space. Solving our segment-closest-point problem is equivalent to locating the cell $C^*$ containing a
point $p^*$ parameterized by the input of our problem (i.e., the segments of $S$ and
the points of $P$) in an arrangement $\calA$ of the boundaries of $\poly(n)$
semialgebraic sets in $O(n)$-dimensional space. This arrangement can be built in
$O(2^{\poly(n)})$ time without examining the values of the input and thus does
not require any comparisons. In particular, the number of cells of $\calA$ is
bounded by $n^{O(n)}$. As a $\Gamma$-algorithm progresses, it maintains a set
$\Pi$ of cells of $\calA$. Initially, $\Pi$ consisting of all cells of $\calA$.
During the course of the algorithm, $\Pi$ can only shrink but always contains
the cell $C^*$. At the end of the algorithm, $C^*$ will be found.
Define the potential $\Phi=\log |\Pi|$. As $\calA$ has $n^{O(n)}$ cells,
initially $\Phi=O(n\log n)$. For any operation or subroutine of the algorithm,
we use $\Delta\Phi$ to denote the change of $\Phi$. As $\Phi$ only decreases
during the algorithm, $\Delta\Phi\leq 0$ always holds and the sum of
$-\Delta\Phi$ during the entire algorithm is $O(n\log n)$. This implies that we
may afford an expensive operation/subroutine during the algorithm as long as it decreases
$\Phi$ a lot.

Two algorithmic tools are developed in \cite{ref:ChanHo22} under the framework: {\em basic search lemma} (Lemma~4.1~\cite{ref:ChanHo22}) and {\em search lemma} (Lemma
A.1~\cite{ref:ChanHo22}). Roughly speaking, given $r$ predicates (each predicate is a test of whether $\gamma(x)$ is true for the input vector $x$), suppose it is promised that at least one of them is true for all inputs in the active cells; then the basic search lemma can find a predicate that is true by making $O(1-r\cdot \Delta\Phi)$ comparisons. Given a binary tree (or a more general DAG of $O(1)$ degree) such that each node $v$ is associated with a predicate $\gamma_v$, suppose for each internal node $v$, $\gamma_v$ implies $\gamma_u$ for a child $u$ of $v$ for all inputs in the active cells. Then, the search lemma can find a leaf $v$ such that $\gamma_v$ is true by making $O(1-\Delta\Phi)$ comparisons.

An application of both lemmas particularly discussed in \cite{ref:ChanHo22} is to find a predecessor of a query number among a sorted list of input numbers. In our algorithm, as will be seen later, the subproblem that needs the $\Gamma$-algorithm framework is also finding predecessors among sorted lists and thus both the basic search lemma and the search lemma are applicable.

\medskip
In the following two subsections, we will tackle the above two challenges,
respectively. By slightly abusing the notation, let $P$ be a set of $n$ points and $S$ a set of $n$ segments for the problem in recurrence~\eqref{equ:50}.

\subsection{Solving outside-hull segment queries}


Recall that we have used the BS data structure to answer the outside-hull segment
queries. Unfortunately the algorithm does not fit into the $\Gamma$-algorithm
framework. Indeed, the BS data structure is a binary tree. However, each node of the tree represents a convex hull of a subset of points and it is not associated with a predicate that we can use to apply the $\Gamma$-algorithm framework (e.g., the search lemma as discussed above).

\subparagraph{Remark.} We provide more details on why the BS data structure does not fit into the $\Gamma$-algorithm
framework. Roughly speaking, in order to fit into the framework, the search structure has to have certain kind of ``locality'' property. Let's use an example to explain this. Suppose we want to find the predecessor of a query number $x$ among a sorted list $X$ of numbers. Assume that the numbers of $X$ are stored in a binary search tree $T_X$. Hence, each node $v$ of $T_X$ is associated with a range formed by the leftmost and rightmost leaves of the subtree rooted at $v$. It is easy to see that if the predecessor of $x$ is in the range of $v$, then it must be in the range of exactly one of the two children of $v$. This is the locality property for the predecessor searching problem, i.e., by looking at the local information at $v$, we are able to determine which subtree of $v$ contains the search target.
For our outside-hull segment query problem, the search algorithm of the BS data structure does not have this locality property. For example, the algorithm searches a binary tree from the root and eventually obtain two candidates along two search paths from the root; finally, the better one from the two candidates is returned as the answer. Therefore, the search target does not only rely on the local information and it has to compare with the result from the other search path (in contrast a locality property always guarantees that only one search path is necessary).

\medskip

In the following, we first present a new algorithm for
solving the outside-hull segment queries. Our algorithm,
whose performance matches that of the BS data structure, is simpler, and thus may be of independent interest; more importantly, it leads to an algorithm that fits the $\Gamma$-algorithm framework to provide an $O(n^{4/3})$ upper bound.

Let $Q$ be a set of $n'$ points. The problem is to preprocess $Q$
so that given any query segment $s$ outside the convex hull $\ch(Q)$ of $Q$, the
closest point of $s$ in $Q$ can be computed efficiently. Recall that in our original problem (i.e., the recurrence~\eqref{equ:50}) $Q$ is a subset of $P$ and the sum of $n'$ for all subsets of $P$ that we need to build the outside-hull query data structures is $O(n\log n)$. We make this an observation below, which will be referred to later.

\begin{observation}\label{obser:20}
The size of the subsets of $P$ that we need to build the outside-hull query data structures is $O(n\log n)$, i.e., $\sum n' =O(n\log n)$.
\end{observation}

In the preprocessing, we compute the Voronoi diagram
$\vd(Q)$ of $Q$, from which we can obtain the convex hull $\ch(Q)$ in linear time. For each edge $e$ of $\ch(Q)$, we determine the subset
$Q_e$ of points of $Q$ whose Voronoi cells intersect $e$ in order
along $e$. This order is exactly the order of the perpendicular
projections of the points of $Q_e$ onto $e$~\cite{ref:BespamyatnikhQu00}.

Consider a query segment $s$ that is outside $\ch(Q)$. Let $p_s$ be the first
point of $Q$ hit by $s$ if we drag $s$ along the direction perpendicularly to
$s$ and towards $\ch(Q)$; see Fig.~\ref{fig:vd}. For ease of exposition, we assume that $p_s$ is unique. Our goal is to compute $p_s$ in the case where
the point of $s$ closest to $Q$ is not an endpoint of $s$ since the other case
can be easily solved by using $\vd(Q)$. Henceforth, we assume that
the point of $s$ closest to $Q$ is not an endpoint of $s$, implying that $p_s$ is the point of $Q$ closest to $s$.
Without loss of generality, we assume that $s$ is horizontal and $s$ is below
$\ch(Q)$.
Let $a$ and $b$ be the left and right endpoints of $s$,
respectively (see Fig.~\ref{fig:vd}).

\begin{figure}[t]
\begin{minipage}[t]{\textwidth}
\begin{center}
\includegraphics[height=2.3in]{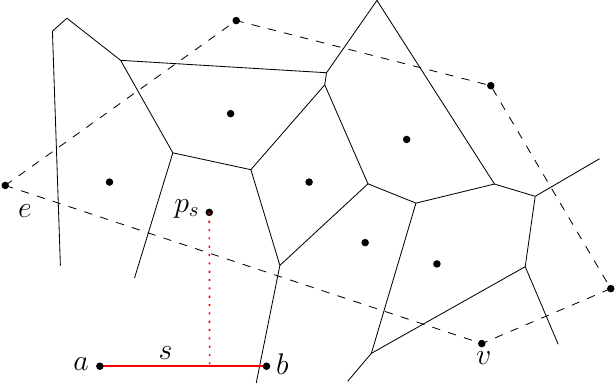}
\caption{\footnotesize Illustrating an outside-hull segment query.
}
\label{fig:vd}
\end{center}
\end{minipage}
\vspace{-0.15in}
\end{figure}

We first find the lowest vertex $v$ of $\ch(Q)$, which can be done in $O(\log n')$
time by doing binary search on $\ch(Q)$. If $x(a)\leq x(v)\leq x(b)$, then $v$ is $p_s$ and we are done with the query.
Otherwise, without loss of generality, we assume that $x(b)<x(v)$. By binary
search on $\ch(Q)$, we find the edge $e$ in the lower hull of $\ch(Q)$ that intersects
the vertical line through $b$. Since  $x(a)\leq x(b)<x(v)$,
$e$ must have a negative slope (see Fig.~\ref{fig:vd}).
Then, as discussed
in~\cite{ref:BespamyatnikhQu00}, $p_s$ must be in $Q_e$. To
find $p_s$ efficiently, we first make some observations (which were not discovered in the previous work).

Suppose $p_1,p_2,\ldots,p_m$ are the points of $Q_e$, sorted following the order of their Voronoi
cells in $\vd(Q)$ intersecting $e$ from left to right. We define two special indices $i^*$ and $j^*$ of $Q_e$ with respect to $a$ and $b$, respectively.
\begin{definition}
Define $j^*$ as the largest index of the point of $Q_e$ that is to the left of $b$.
Define $i^*$ as the smallest index of the point of $Q_e$ such that
$p_j$ is to the right of $a$ for all $j\geq i^*$.
\end{definition}
Note that $j^*$ must exist as $p_s$ is in $Q_e$ and is
to the left of $b$.
We have the following lemma.

\begin{lemma}\label{lem:20}
If $i^*$ does not exist or $i^*>j^*$, then $p_s$ cannot be the closest point of $s$ in $Q$.
\end{lemma}
\begin{proof}
We first assume that $i^*$ exists and $i^*>j^*$. In the following, we prove that $p_s$ cannot be the closest point of $s$ in $Q$.
By the definition of $j^*$, for any $i\geq
j^*+1$, $p_i$ is to the right of $b$ and thus is to the right of $a$.
Hence, by the definition of $i^*$, it holds that $i^*\leq j^*+1$.
Since $i^*>j^*$, $i^*$ must be $j^*+1$.
Again by the definition of $i^*$, $p_{j^*}$ must be strictly to the left
of $a$ since otherwise $i^*\leq j^*$ must hold.

\begin{figure}[t]
\begin{minipage}[t]{\textwidth}
\begin{center}
\includegraphics[height=1.5in]{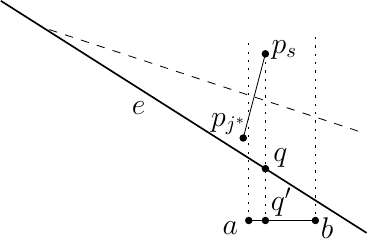}
\caption{\footnotesize Illustrating the proof of Lemma~\ref{lem:20}. The point $p_i$ is $p_s$. The dashed line between $p_{j^*}$ and $p_s$ is their bisector.
}
\label{fig:fig10}
\end{center}
\end{minipage}
\vspace{-0.15in}
\end{figure}

Let $q$ (resp., $q'$) be the intersections between $e$ (resp., $s$)
and the vertical line through
$p_s$ (e.g., see Fig.~\ref{fig:fig10}). Note that both $q$ and $q'$ must exit by the definition of
$p_s$.

Assume to the contradiction that $p_s$ is the closest point of $s$ in $Q$. Without loss of
generality, assume $p_s$ is $p_i$ for some index $i\in [1,m]$. Then,
since $x(a)\leq x(p_i)\leq x(b)$, we have $i\leq j^*$ by the
definition of $j^*$. Further, $i$ cannot be $j^*$ since $p_{j^*}$ is strictly
to the left of $a$. Therefore, $i<j^*$.
Recall that the index order of $Q_e$ follows the intersections of
the Voronoi cells with $e$ from left to right. Since $i<j^*$ and
$x(p_{j^*})<x(a)\leq x(p_i)$, the bisector of $p_i$ and $p_{j^*}$ intersects
$e$ at a point to the left of $a$ (e.g., see Fig.~\ref{fig:fig10}).
As $q$ is to the right of $a$, $q$ must be closer to $p_{j^*}$ than to $p_i$.

On the other hand, since $p_i$ is the closest point of $s$ in $Q$ and
$\overline{p_iq'}$ is perpendicular to $s$, for any point $p\in
\overline{p_iq'}$, $p$'s closest point in $Q$ is $p_i$. Hence, $q$, which is a
point on $\overline{p_iq'}$,
must be closer to $p_i$ than to $p_{j^*}$. We thus obtain contradiction.

We next argue that if $i^*$ does not exist, then $p_s$ cannot be the closest point of $s$ in $Q$. The argument is similar as above. First notice that $p_m$ must be strictly to the left of $a$, since otherwise $i^*$ would exit. Without loss of generality, assume $p_s$ is $p_i$ for some index $i\in [1,m]$. Since $x(a)\leq x(p_i)\leq x(b)$, $i$ cannot be $m$ and thus $i<m$. Hence, we have $i<m$ and $x(p_m)<x(a)\leq x(p_i)$. Then, by applying the same argument as above (just replace $p_{j^*}$ by $p_m$), we can prove that $p_s$ cannot be the closest point of $s$ in $Q$.
\end{proof}

By Lemma~\ref{lem:20}, if $i^*$ does not exist or if $i^*>j^*$, then we can simply stop the query algorithm.
In the following, we assume that $i^*$ exists and $i^*\leq j^*$.
Let $Q_e[i^*,j^*]$ denote the subset of points of $Q_e$ whose indices
are between $i^*$ and $j^*$ inclusively. The following lemma implies that we can
use the supporting line of $s$ to search $p_s$.

\begin{lemma}\label{lem:30}
Suppose $p_s$ is the closest point of $s$ in $Q$. Then, $p_s$ is the
point of $Q_e[i^*,j^*]$ closest to the supporting line of $s$ (i.e., the line containing $s$).
\end{lemma}
\begin{proof}
Let $\ell$ denote the supporting line of $s$ and let $p^*$ denote the
point of $Q_e[i^*,j^*]$ closest to $\ell$. Our goal is to prove that
$p^*$ is $p_s$, i.e., $p^*$ is the closest point of $s$ in $Q$.
To this end, it suffice to prove the following: (1) $x(a)\leq x(p^*)\leq x(b)$;
(2) for any point $p$ of $Q_e$ not in
$Q_e[i^*,j^*]$,
$p$ cannot be $p_s$.

We first prove (1).
Consider a point $p_i\in Q_e[i^*,j^*]$ such that $x(p_i)\not\in [x(a),x(b)]$.
In the following we prove that $Q_e[i^*,j^*]$ must have another point $p_j$ such that
$d(p_j,\ell)<d(p_i,\ell)$ and $x(p_j)\in [x(a),x(b)]$. This will
lead to (1).


\begin{figure}[t]
\begin{minipage}[t]{0.48\textwidth}
\begin{center}
\includegraphics[height=1.5in]{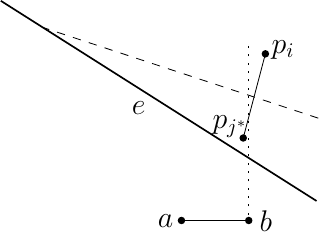}
\caption{\footnotesize The dashed line between $p_{j^*}$ and $p_i$ is the bisector of the two points.}
\label{fig:fig20}
\end{center}
\end{minipage}
\hspace{0.05in}
\begin{minipage}[t]{0.48\textwidth}
\begin{center}
\includegraphics[height=1.5in]{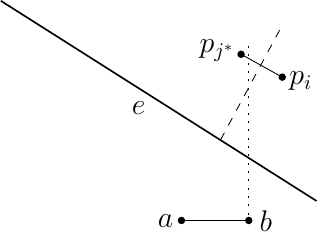}
\caption{\footnotesize The dashed line between $p_{j^*}$ and $p_i$ is the bisector of the two points.}
\label{fig:fig25}
\end{center}
\end{minipage}
\end{figure}

Since $i\in [i^*,j^*]$, $p_i$ is to the right of $a$.
As $x(p_i)\not\in [x(a),x(b)]$, $p_i$ must be strictly to the right of $b$.  Since $p_{j^*}$
is to the left of $b$, we obtain $x(p_{j^*})\leq x(b)<x(p_i)$ and $i<j^*$.
Recall that the index order of $Q_e$ follows the intersections of
the Voronoi cells of $Q_e$ with $e$ from left to right. Since $i<j^*$, the portion of $e$ closer to $p_i$ is to the left of
the portion of $e$ closer to $p_{j^*}$.
This is possible only if $y(p_i)>y(p_{j^*})$ (e.g., see Fig.~\ref{fig:fig20}).
Indeed, assume to the contradiction that $y(p_i)\leq y(p_{j^*})$ (e.g., see Fig.~\ref{fig:fig25}). Then,
since the slope of $e$ is negative and $x(p_{j^*})\leq x(b)< x(p_i)$, the portion of $e$ closer to $p_i$ must be to the right of
the portion of $e$ closer to $p_{j^*}$, incurring contradiction. As such, $y(p_i)>y(p_{j^*})$ must hold.
Hence,
$d(p_{j^*},\ell)<d(p_i,\ell)$ as $\ell$ is horizontal.
Notice that $x(p_{j^*})\in [x(a),x(b)]$. Indeed, by definition,
$p_{j^*}$ is to the left of $b$. On the other hand, since $i^*\leq j^*$,
$p_{j^*}$ is to the right of $a$. Hence, $x(p_{j^*})\in [x(a),x(b)]$. This proves (1)
since $p_{j^*}$ is a point in $Q_e[i^*,j^*]$.

We now prove (2). Consider any point $p_i$ of $Q_e\setminus
Q_e[i^*,j^*]$. Our goal is to prove that $p_i$ is not $p_s$.
Recall that $p_s$ is the closest point of $Q$ and $x(p_s)\in
[x(a),x(b]$. First of all, if $x(p_i)\not\in [x(a),x(b)]$, then it is vacuously
true that $p_i\neq p_s$. We now assume that $x(p_i)\in [x(a),x(b)]$.
Hence, $p_i$ is to the left of $b$. By the definition of $j^*$,
$i\leq j^*$ holds. As $i\not\in [i^*,j^*]$, we have $i<i^*$.
Since $p_i$ is to the right of $a$, by the definition of $i^*$,
there must be a point $p_j$ with $i<j<i^*$ such that $p_j$ is strictly
to the left of $a$. Recall that the index order of $Q_e$ follows the intersections of
the Voronoi cells of $Q_e$ with $e$ from left to right.
Since $i<j$ and $x(p_j)<x(a)\leq x(p_i)$, the bisector of $p_i$ and
$p_j$ intersects $e$ at a point to the left of $a$ (e.g., see Fig.~\ref{fig:fig30}). This means that
$q$ is closer to $p_j$ than to $p_i$, where $q$ is the intersection
between $e$ and the vertical line through $p_i$.

\begin{figure}[t]
\begin{minipage}[t]{\textwidth}
\begin{center}
\includegraphics[height=1.7in]{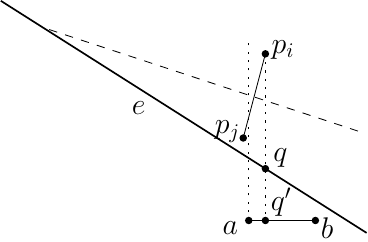}
\caption{\footnotesize The dashed line between $p_j$ and $p_i$ is the bisector of the two points.
}
\label{fig:fig30}
\end{center}
\end{minipage}
\vspace{-0.15in}
\end{figure}

Assume to the contrary that $p_i$ is $p_s$. Let $q'$ be the
intersection between $s$ and the vertical line through $p_i$ (e.g., see Fig.~\ref{fig:fig30}). Since
$p_i$ is the closest point of $s$ in $Q$, every point of
$\overline{q'p_i}$ has $p_i$ as its closest point in $Q$. In
particular, $q$, which is on $\overline{q'p_i}$, is closer to $p_i$ than to $p_j$. But this incurs
contradiction. This proves (2).
\end{proof}

Based on Lemma~\ref{lem:30}, we have the following three steps to
compute $p_s$: (1) compute $j^*$; (2) compute $i^*$; (3) find the point of
$Q_e[i^*,j^*]$ closest to the supporting line $\ell_s$ of $s$.

The following Lemma~\ref{lem:40}, which is for outside-hull segment queries, is a by-product of our above observations. Its complexity is the same as that in \cite{ref:BespamyatnikhQu00}. However, we feel that our new query algorithm is simpler and thus this result may be interesting in its own right. The query algorithm of
the lemma actually does not fit the $\Gamma$-algorithm framework. Instead, following the above observations
we will give another query algorithm that fits the $\Gamma$-algorithm framework.

\begin{lemma}\label{lem:40}
Given a set $Q$ of $n'$ points in the plane, we can build a data structure of
$O(n')$ space in $O(n'\log n')$ time such that each outside-hull query can be
answered in $O(\log n')$ time. The preprocessing time is $O(n')$ if the Voronoi
diagram of $Q$ is known.
\end{lemma}
\begin{proof}
The preprocessing algorithm is essentially the same as that in~\cite{ref:BespamyatnikhQu00}.
We first compute the Voronoi diagram
$\vd(Q)$ of $Q$, from which we can obtain the convex hull $\ch(Q)$ of $Q$ in linear time. Then, we determine $Q_e$ for each edge $e$ of
$\ch(Q)$. We preprocess each $Q_e$ as follows.
We build a balanced binary search tree $T_e$ whose leaves corresponding to
the points of $Q_e=\{p_1,p_2,\ldots,p_m\}$
in their index order as discussed before. For each node $v$ of $T_e$, we use
$Q_e(v)$ to denote the set of points in the leaves of the subtree
rooted at $v$. Before enhancing $T_e$ with additional information, we describe
our query algorithm.

Consider a query segment $s$. Without loss of generality, we assume that $s$ is
horizontal and below $\ch(Q)$. Let $a$ and $b$ be the left and right endponits
of $s$, respectively. As discussed before,
we first find the lowest vertex $v$ of $\ch(Q)$, which can be done in $O(\log n)$
time by doing binary search on $\ch(Q)$. If $x(a)\leq x(v)\leq x(b)$, then $v$
must be the closest point of $Q$ to $s$ and we are done with the query.
Otherwise, without loss of generality, we assume that $x(b)<x(v)$. By binary
search on $\ch(Q)$, we find in $O(\log n)$ time
the edge $e$ in the lower hull of $\ch(Q)$ that intersects
the vertical line through $b$. Then, the closest point of $s$ in $Q$ must be in
$Q_e$~\cite{ref:BespamyatnikhQu00}.
We next find $p_s$ using $T_e$ and the algorithm has three steps as
discussed above.

\begin{enumerate}
\item
First, for computing $j^*$, starting from the
root of $T_e$, for each node $v$, we do the following. Let $u$ be
the right child of $v$. Let $x_{l}$ and $x_r$ denote the
$x$-coordinates of the leftmost and rightmost points of $Q_e(u)$, respectively.
If $x_l\leq x(b)$, then $p_{j^*}$ is in $Q_e(u)$ and we proceed on
$u$. Otherwise, $p_{j^*}$ cannot be in $Q_e(u)$ and we proceed on the
left child of $v$.

\item
Second, for computing $i^*$, starting from the
root of $T_e$, for each node $v$, we do the following. Let $u$ be
the right child of $v$. Let $w$ be the left child of $v$. Let $x_{l}$ and $x_r$ denote the
$x$-coordinates of the leftmost and rightmost points of $Q_e(u)$, respectively.
If $x_l< x(a)\leq x_r$, then $p_{i^*}$ must be in $Q_e(u)$ and we
proceed on $u$. If $x(a)\leq x_l$, then we proceed on $w$ since $p_{i^*}$ is either in $Q_e(w)$ or
the leftmost leaf of the subtree rooted at $u$ (the latter case will be handled
next). If $x(a)>x_r$, then $p_{i^*}$ is $w''$, where
$w''$ is the right neighboring leaf of $w'$ and $w'$ is the rightmost leaf of
the subtree rooted at $u$.

\item
After $i^*$ and $j^*$ are found, by standard approach, we can obtain a set $V$ of
$O(\log n)$ nodes of
$T_e$ such that the union of $Q_e(v)$ of all nodes $v\in V$ is exactly
$Q_e[i^*,j^*]$. For each node $v\in V$, we find the lowest point of
$Q_e(v)$ as a candidate; finally among all such candidate points, we return the lowest one as $p_s$.
\end{enumerate}

To implement all above three steps in $O(\log n')$ time, we enhance $T_e$
in the same way as that in~\cite{ref:BespamyatnikhQu00}. We briefly
discuss it here for completeness. We first
store the convex hull $\ch(Q_e)$ at the root of $T_e$. Any other internal node
stores the portion of its convex hull that is not stored by its
ancestors. For this, a key observation is that the convex hull of a
node can be obtained from the convex hulls of its children by
computing the two common tangents (this is because the two subsets of
points at the two children are separated by a line perpendicular to $e$).
In addition, we construct a fractional cascading
structure~\cite{ref:ChazelleFr86} so that if a tangent to the convex
hull at a node $v$
is known, then the tangents of the same slope to the convex hulls of
the two children can be found in constant time. The total
preprocessing time for $Q$ is $O(n'\log n')$ and the space is $O(n')$. In
addition, if the Voronoi diagram of $Q$ is known, then the
preprocessing time can be reduced to $O(n')$.
In this way, all above three steps (and thus the entire query algorithm)
can be implemented in $O(\log n')$
time. For example, to compute $j^*$, we need to access $O(\log n')$
nodes and for each such node $v$, we need to find the leftmost and
rightmost points of $Q_e(v)$. This can be done in $O(\log n')$ time
using the fractional cascading structure by computing the tangents to
their convex hulls of a slope perpendicular to $s$.

In summary, the above gives an $O(n')$ space data structure that can
support each outside-hull segment query in $O(\log n')$ time. The
data structure can be built in $O(n'\log n')$ time or in $O(n')$ time if
the Voronoi diagram of $Q$ is known.
\end{proof}

We now give a new algorithm that fits the $\Gamma$-algorithm framework.
The new algorithm requires slightly more preprocessing than Lemma~\ref{lem:40}.
But for our purpose, we are satisfied with $O(n^{4/3})$ preprocessing time.
We have different preprocessing for each of the three steps of the query algorithm, as follows.

\subparagraph{The first step: computing $\boldsymbol{j^*}$.}
For computing $j^*$, we will use the {\em basic search lemma} (i.e.,
Lemma~4.1) in~\cite{ref:ChanHo22}. In order to apply the lemma,
we perform the following preprocessing.

Recall that $Q_e=\{p_1,p_2,\ldots,p_m\}$ is ordered by their Voronoi cells
intersecting $e$. We partition the sequence into $r$ contiguous subsequences of
size roughly $m/r$ each. Let $Q_e^i$ denote the $i$-th subsequence, with $1\leq i\leq r$.
For each $i\in [1,r]$, we compute and explicitly
maintain the convex hull $\ch(i)$ of all points in the union
of the subsequences $Q_e^j$, $j=i,i+1,\ldots,r$.
Next, for each subsequence $Q_e^i$, we further partition it into $r$ contiguous sequences
of size roughly $|Q_e^i|/r$ and process it in the same way as above. We do
this recursively until the subsequence has no more than $r$ points. In this
way, we obtain a tree $T$ with $m$ leaves such that each node has $r$ children. For each node $v$, we use $\ch(v)$ to denote the convex hull that is computed above corresponding to $v$ (e.g., if $v$ is the child of the root corresponding to $Q_e^i$, then $\ch(v)$ is $\ch(i)$ defined above).
The total time for constructing $T$ can be easily bounded by $O(mr\log m\log_r
m)$ as the height of $T$ is $O(\log_rm)$.

Now to compute $j^*$, we search the tree $T$: starting from
the root, for each node $v$, we apply the basic search lemma on all $r$
children of $v$. Indeed, this is possible due to the following. Consider the
root $v$. For each $i$ with $1\leq i\leq r$, let $x_i$ denote the
$x$-coordinate of the leftmost point of the union of the subsequences
$Q_e^j$, $j=i,i+1,\ldots,r$; note that $x_i$ is also the leftmost vertex of $\ch(i)$.
It is not difficult to see that $x_1\leq x_2\leq \ldots\leq
x_{r}$. Observe that $p_{j^*}$ is in $Q_e^i$ if and only if $x_i\leq x(b)<x_{i+1}$.
Therefore, we find the index $i$ such that $x_i\leq x(b)<x_{i+1}$ and then
proceed to the child of $v$ corresponding to $Q_e^i$. This property satisfies the
condition of the basic search lemma (essentially, we are looking for the
predecessor of $b$ in the sequence $x_1,x_2,\ldots,x_{r}$ and this is somewhat
similar to the insertion sort algorithm of Theorem~4.1~\cite{ref:ChanHo22}, which uses the basic search lemma). By
the basic search lemma, finding the index $i$ can be done using $O(1-r\Delta\Phi)$
comparisons provided that the $x$-coordinates $x_1,x_2,\ldots,x_{r}$ are available to
us (we will discuss how to compute them later). We then follow the same idea recursively until we reach a leaf.
In this way, the total number of comparisons for computing $j^*$ is
$O(\log_rm-r\Delta\Phi)$.

By setting $r=m^{\epsilon}$ for a small constant $\epsilon$, the
preprocessing time is $O(m^{1+\epsilon}\log m)$ and
computing $j^*$ can be done using $O(1-m^{\epsilon}\Delta\Phi)$ comparisons.
Recall that there are $O(n^{4/3})$ queries in our original problem (i.e., recurrence~\eqref{equ:50}) and the total time for $-\Delta\Phi$ during the entire algorithm is $O(n\log n)$. Also, since $m$ is the number of points of $Q$ whose Voronoi cells intersecting the edge $e$ of $\ch(Q)$,
the sum of $m$ for all outside-hull segment query data structures for all edges of $\ch(Q)$ is $|Q|$, which is $n'$. By Observation~\ref{obser:20}, the sum of $n'$ for all data structures in our original problem is $O(n\log n)$.
Hence, the total preprocessing time for our original problem is $O(n^{1+\epsilon}\log^{2+\epsilon}n)$, which is bounded by $O(n^{4/3})$ if we set $\epsilon$ to a small constant (e.g.,
$\epsilon=1/4$). As such,
with a preprocessing step of $O(n^{4/3})$ time, we can compute $j^*$ for all
queries using a total of $O(n^{4/3})$ comparisons.

The above complexity analysis for computing $j^*$ is based on the assumption that the leftmost
point of $\ch(v)$ for each node $v$ of $T$ is known. To find these points during
the queries, we take advantage of the property that all queries are
offline, i.e., we know all query segments before we start the queries.
Notice that although there are
$O(n^{4/3})$ queries, the number of distinct query segments is $n$, i.e., those
in $S$ (a segment may be queried on different subsets of $P$).
Let $s$ be the current query segment and $p$ be the leftmost point of a convex
hull $CH(v)$ with respect to $s$ (i.e., by assuming $s$ is horizontal). Let
$\rho_1$ be the ray from $p$ going vertically upwards. Let $\rho_2$ be another ray from $p$
going through the clockwise neighbor of $p$ on $CH_v$, i.e., $\rho_2$ contains
the clockwise edge of $CH_v$ incident to $v$. Observe that for another query
segment $s'$, $p$ is still the leftmost point of $CH_v$ with respect to $s'$ as
long as the direction perpendicular to $s'$ is within the angle from $\rho_1$
clockwise to $\rho_2$.
Based on this observation, before we start any query, we sort the perpendicular
directions of all segments of $S$ along with the directions of all edges of all
convex hulls of all nodes of the trees $T$ for all outside-hull segment query
data structures in our original problem (i.e., the recurrence~\eqref{equ:50}).
As analyzed above, the total size of convex hulls of all trees $T$ is
$O(n^{1+\epsilon}\log^{2+\epsilon} n)$. Hence, the sorting can be done in
$O(n^{1+\epsilon}\log^{3+\epsilon} n)$ time. Let $L$ be the sorted list. We
solve the queries for segments following their order in $L$.
Let $s$ and $s'$ be two consecutive segments of $S$ in $L$. After we solve all
queries for $s$, the directions between $s$ and $s'$ in $L$ correspond to those
nodes of the trees $T$ whose leftmost points need to get updated, and we then
update the leftmost points of those nodes before we solve queries for $s'$. The
total time we update the tree nodes for all queries is proportional to the total
size of all trees, which is $O(n^{1+\epsilon}\log^{2+\epsilon} n)$.

In summary, after $O(n^{4/3})$ time preprocessing, computing $j^*$ for all
$O(n^{4/3})$ outside-hull segment queries can be done using $O(n^{4/3})$
comparisons.

\subparagraph{The second step: computing $\boldsymbol{i^*}$.}
For computing $i^*$, the idea is similar and we only sketch it.
In the preprocessing, we build the same tree $T$ as above for the first step.
One change is that we add the first point $p$ of the subsequence $Q_e^i$ to the
end of $Q_e^{i-1}$, i.e., $p$ appears in both $Q_e^i$ and $Q_e^{i-1}$. This does
not change the complexities asymptotically.

For each query, to compute $i^*$, consider the root $v$. Observe that $i^*$ is
in $Q_e^i$ if and only if $x_i<x(a)\leq x_{i+1}$ ($x_i$ and $x_{i+1}$ are
defined in the same way as before). As such, we can apply the basic search lemma to
find $i^*$ in $O(1-m^{\epsilon}\Delta\Phi)$ comparisons. We can use the same
approach as above to update the leftmost points of convex hulls of nodes of the
trees $T$ (i.e., computing a sorted list $L$ and process the queries of the
segments following their order in $L$).

In summary, after $O(n^{4/3})$ time preprocessing, computing $i^*$ for all
$O(n^{4/3})$ outside-hull segment queries can be done using $O(n^{4/3})$
comparisons.

\subparagraph{The third step.}
The third step of the query algorithm is to find the point $p_s$ of
$Q_e[i^*,j^*]$ closest to $\ell(s)$, where $\ell(s)$ is the supporting line of $s$.
We first discuss the preprocessing step on $Q_e$.

We build a balanced binary search tree $T_e$ whose leaves corresponding to
the points of $Q_e=\{p_1,p_2,\ldots,p_m\}$
in their index order as discussed before. For each node $v$ of $T_e$, we use
$Q_e(v)$ to denote the set of points in the leaves of the subtree
rooted at $v$.
For each node $v$ of $T_e$, we explicitly store the convex hull of
$Q_e(v)$ at $v$.
Further, for each leaf $v$, which stores a point $p_i$ of $Q_e$, for each ancestor
$u$ of $v$, we compute the convex hull $\ch_r(v,u)$ of all points
$p_i,p_{i+1},\ldots,p_j$, where $p_j$ is the point in the
rightmost leaf of the subtree at $u$. We do this in a bottom-up manner
starting from $v$ following the path from $v$ to $u$. More specifically,
suppose we are currently at a node $w$, which is $v$ initially.
Suppose we have the convex hull $\ch_r(v,w)$.
We proceed on the parent $w'$ of $w$ as follows. If $w$ is the right child of $w'$, then $\ch_r(v,w')$ is $\ch_r(v,w)$ and thus we do nothing.
Otherwise, we merge $\ch_r(v,w)$ with the convex hull of $Q_e(w'')$ at $w''$, where
$w''$ is the right child of $w'$. Since points of $\ch_r(v,w)$ are
separated from points of $Q_e(w'')$ by a line perpendicular to
$e$~\cite{ref:BespamyatnikhQu00}, we can merge the two hulls by
computing their common tangents in $O(\log m)$
time~\cite{ref:OvermarsMa81}. We use a persistent tree to maintain
the convex hulls (e.g., by a path-copying
method)~\cite{ref:DriscollMa89,ref:SnarnakPl86} so that after the
merge we still keep $\ch_r(v,w)$. In this way, we have
computed $\ch_r(v,w')$ and we then proceed on the parent
of $w'$. We do this until we reach the root. As such, the total time and
extra space for computing the convex hulls for a leaf $v$ is $O(\log^2
m)$, and the total time and space for doing this for all leaves
is $O(m\log^2 m)$. Symmetrically, for each leaf $v$, which stores a point $p_i$ of $Q_e$, for each ancestor $u$ of $v$, we compute the convex hull $\ch_l(v,u)$ of all points
$p_h,p_{h+1},\ldots,p_i$, whether $p_h$ is the point in the
leftmost leaf of the subtree at $u$. Computing the convex hulls $\ch_l(v,u)$ for all ancestors $u$ for all leaves $v$ can be done in $O(m\log^2 m)$ in a similar way as above.
In addition, we construct a lowest common ancestor (LCA) data
structure on the tree $T_e$ in $O(m)$ time so that the LCA of any two query nodes of
$T_e$ can be found in $O(1)$ time~\cite{ref:BenderTh00,ref:HarelFa84}.
The total preprocessing time for constructing the tree $T_e$ as above is $O(m\log^2 m)$.
Recall that the sum of $m$ for all outside-hull segment query data structures is $O(n\log n)$. Therefore, the total preprocessing time as above for all data structures is $O(n\log^3 n)$.

Now consider the third step of the query algorithm. Suppose
$i^*$ and $j^*$ are known. The problem is to compute the point $p_s$ of
$Q_e[i^*,j^*]$ closest to the supporting line $\ell(s)$ of $s$. Let
$u$ and $v$ be the two leaves of $T_e$ storing the two points $p_{i^*}$
and $p_{j^*}$, respectively. Let $w$ be the lowest common ancestor of $u$
and $v$. Let $u'$ and $v'$ be the left and right children of $w$,
respectively. It is not difficult to see that the convex hull of
$\ch_r(u,u')$ and $\ch_l(v,v')$ is the convex hull of $Q_e[i^*,j^*]$. As
such, to find $p_s$, it suffices to compute the vertex of
$\ch_r(u,u')$ closest to $\ell(s)$ and the vertex of $\ch_l(v,v')$ closest to
$\ell(s)$, and among the two points, return the one closer to $\ell(s)$ as $p_s$.
To implement the algorithm, finding $w$ can be done in $O(1)$ time
using the LCA data structure~\cite{ref:BenderTh00,ref:HarelFa84}.
To find the closest vertex of $\ch_r(u,u')$ to $\ell(s)$, recall that the preprocessing computes a balanced binary search tree (maintained by a persistent tree), denoted by $T_r(u,u')$, for maintaining $\ch_r(u,u')$.
We apply a {\em search lemma} of Chan and Zheng (Lemma
A.1~\cite{ref:ChanHo22}) on the tree $T_r(u,u')$. Indeed, the problem is
equivalent to finding the predecessor of the slope of $\ell(s)$ among the slopes
of the edges of $\ch_r(u,u')$. Using the search lemma, we can find the
vertex of $\ch_r(u,u')$ closest to $\ell(s)$ using $O(1-\Delta\Phi)$ comparisons.
Similarly, the vertex of $\ch_l(v,v')$ closest to $\ell(s)$ can be found
using $O(1-\Delta\Phi)$ comparisons. In this way, $p_s$ can be computed using $O(1-\Delta\Phi)$ comparisons.

In summary, with $O(n\log^3 n)$ time preprocessing, the third step of the query algorithm for all
$O(n^{4/3})$ queries can be done using a total of $O(n^{4/3})$ comparisons (recall that the sum of
$-\Delta\Phi$ in the entire algorithm is $O(n\log n)$).

\subparagraph{Summary.}
Combining the three steps discussed above, 
all $O(n^{4/3})$ outside-hull segment queries can be solved using $O(n^{4/3})$ comparisons.
Recall that the above only discussed
the query on the data structure for a single edge $e$ of the convex hull of $Q$.
As the first procedure of the query, we need to find the vertex of $\ch(Q)$ closest to
the supporting line of $s$.
For this, we can maintain the convex hull $\ch(Q)$ by a balanced binary search tree
and apply the search lemma of Chan and Zheng (Lemma
A.1~\cite{ref:ChanHo22}) in the same way as discussed above. As such, this procedure for all
queries uses $O(n^{4/3})$ comparisons. The second procedure of the query is to find the edge of
$\ch(Q)$ intersecting the line through one of the endpoints of $s$ and perpendicular to
$s$. This operation is essentially to find a predecessor of the above endpoint of $s$ on the vertices of the lower hull of $\ch(Q)$. Therefore, we
can also apply the search lemma of Chan and Zheng, and thus this procedure for all
queries also uses $O(n^{4/3})$ comparisons.
As such, we can solve all $O(n^{4/3})$ outside-hull segment queries using $O(n^{4/3})$ comparisons, or
alternatively, we have an algebraic decision tree of height $O(n^{4/3})$ that
can solve all $O(n^{4/3})$ queries.


\subsection{Solving the subproblems $\boldsymbol{T(n^{2/3},n^{1/3})}$}

We now tackle the second challenge, i.e., solve each subproblem $T(n^{2/3},n^{1/3})$ in recurrence~\eqref{equ:50} using $O(n^{2/3})$ comparisons, or solve all $O(n^{2/3})$ subproblems $T(n^{2/3},n^{1/3})$ in~\eqref{equ:50} using $O(n^{4/3})$ comparisons.

Recall that $P$ is the set of $n$ points and $S$ is the set of $n$ segments for the original problem in recurrence~\eqref{equ:50}.
If the closest point of a segment $s\in S$ to $P$ is an endpoint of $s$, then finding the closest point of $s$ in $P$ can be easily done using the Voronoi diagram of $P$. Hence, it suffices to find the first point of $P$ hit by $s$ if we drag $s$ along the directions perpendicularly to $s$. There are two such directions, but in the following discussion we will only consider dragging $s$ along the upward direction perpendicularly to $s$ (recall that $s$ is not vertical due to our general position assumption) and let $p_s$ be the first point of $P$ hit by $s$, since the algorithm for the downward direction is similar. As such, the goal is to compute $p_s$ for each segment $s\in S$.

For notational convenience, let $m=n^{1/3}$ and thus we want to
solve $T(m^2,m)$ using $O(m^2)$ comparisons. More specifically, we are given $m$ points
and $m^2$ segments; the problem is to compute for each segment $s$ the point $p_s$ (with respect to the $m$ points, i.e., the first point hit by $s$ if we drag $s$ along the upward direction perpendicular to $s$). Our goal is to solve all $O(m^2)$ segment dragging queries using $O(m^2)$ comparisons after certain preprocessing. In what follows, we begin with the preprocessing algorithm.

\subparagraph{Preprocessing.}
For two sets $A=\{a_1,a_2,\ldots,a_m\}$ and $B=\{b_1,b_2,\ldots,b_m\}$ of $m$ points each, we say that they have the same {\em order type} if for each $i$, the index order of the points of $A$ sorted around $a_i$ is the same as that of the points of $B$ sorted around $b_i$ (equivalently, in the dual plane, the index order of the dual lines intersecting the dual line of $a_i$ is the same as that of the dual lines intersecting the dual line of $b_i$); the concept has been used elsewhere, e.g.,~\cite{ref:AronovSu21,ref:ChanHo22,ref:GoodmanMu83}. Because constructing the arrangement of a set of $m$ lines can be computed in $O(m^2)$ time~\cite{ref:ChazelleTh85}, we can decide whether two sets $A$ and $B$ have the same order type in $O(m^2)$ time, e.g., simply follow the incremental line arrangement construction algorithm~\cite{ref:ChazelleTh85}. We actually build an algebraic decision tree $T_D$ so that each node of $T_D$ corresponds to a comparison of the algorithm. As such, the height of $T_D$ is $O(m^2)$ and $T_D$ has $2^{O(m^2)}$ leaves, each of which corresponds to an order type (note that the number of distinct order types is at most $m^{6m}$~\cite{ref:GoodmanUp86}, but here using $2^{O(m^2)}$ as an upper bound suffices for our purpose).

Let $Q$ be a set of $m$ points whose order type corresponds to a leaf $v$ of $T_D$.
Let $K_Q$ denote the set of the slopes of all lines through pairs of points of $Q$. Note that $|K_Q|=O(m^2)$. We sort the slopes of $K_Q$. Consider two consecutive slopes $k_1$ and $k_2$ of the sorted $K_Q$. In the dual plane, for any vertical line $\ell$ whose $x$-coordinate is between $k_1$ and $k_2$, $\ell$ intersects the dual lines of all points of $Q$ in the same order (because $k_1$ and $k_2$ respectively are $x$-coordinates of two consecutive vertices of the arrangement of dual lines). This implies the following in the primal plane. Consider any two lines $\ell_1$ and $\ell_2$ whose slopes are between $k_1$ and $k_2$ such that all points of $Q$ are above $\ell_i$ for each $i=1,2$. Then, the order of the lines of $Q$ by their distances to $\ell_1$ is the same as their order by the distances to $\ell_2$. However, if we project all points of $Q$ onto $\ell_1$ and $\ell_2$, the orders of their projections along the two lines may not be the same. To solve our problem, we need a stronger property that the above projection orders are also the same. To this end, we further refine the order type as follows.

For each pair of points $q_i$ and $q_j$ of $Q$, we add the slope of the line perpendicular to the line through $p$ and $q$ to $K_Q$. As such, the size of $K_Q$ is still $O(m^2)$. Although $K_Q$ has $O(m^2)$ values, all these values are defined by the $m$ points of $Q$. Using this property, $K_Q$ can be sorted using $O(m^2)$ comparisons~\cite{ref:ChanHo22,ref:FredmanHo76}.

For two sets $Q=\{q_1,q_2,\ldots,q_m\}$ and $Q'=\{q_1',q_2',\ldots,q_m'\}$ of
$m$ points each with the same order type, we say that they have the same {\em
refined order type} if the order of $K_Q$ is the same as that of $K_{Q'}$, i.e.,
the slope of the line through $q_i$ and $q_j$ (resp., the slope of the line
perpendicular to the line through $q_i$ and $q_j$) is in the $k$-th position of
the sorted list of $K_Q$ if and only if the slope of the line through $q_i'$ and
$q_j'$ (resp., the slope of the line perpendicular to the line through $q_i'$
and $q_j'$) is in the $k$-th position of the sorted list of $K_{Q'}$. We further
enhance the decision tree $T_D$ by attaching a new decision tree at each leaf
$v$ of $T_D$ for sorting $K_Q$ (recall that $K_Q$ can be sorted using $O(m^2)$
comparisons, i.e., there is an algebraic decision tree of height $O(m^2)$ that
can sort $K_Q$), where $Q$ is a set of $m$ points whose order type corresponds
to $v$. We still use $T_D$ to refer to the new tree. The height of $T_D$ is
still $O(m^2)$.

We perform the following preprocessing work for each leaf $v$ of $T_D$.
Let $Q$ be a set of $m$ points that has the refined order type of $v$.
We associate $Q$ with $v$,
compute and sort $K_Q$, and store the sorted list using a balanced
binary search tree.
Let $k_1$ and $k_2$ be two consecutive slopes in the sorted list of $K_Q$.
Consider a line $\ell$ whose slope is in $(k_1,k_2)$ such that $\ell$ is below all points of $Q$. We project all points
perpendicularly onto $\ell$. According to the definition of $K_Q$, the order of
the projections is fixed for all such lines $\ell$ whose slopes are in
$(k_1,k_2)$. Without loss of generality, we assume that $\ell$ is horizontal.
Let $q_1,q_2,\dots,q_m$ denote the points of $Q$ ordered by their projections on
$\ell$ from left to right and we maintain the sorted list in a balanced binary
search tree. For each pair $(i,j)$ with $1\leq i\leq j\leq m$, let
$Q[i,j]=\{q_i,q_{i+1},\ldots,q_j\}$; we sort all points of $Q[i,j]$ by their distances to $\ell$ and store the sorted list in a balanced binary
search tree. As such, the time we spent on the preprocessing at $v$ is
$O(m^5\log m)$.

Since $T_D$ is a decision tree of height $O(m^2)$, the number of leaves of $T_D$
is $2^{O(m^2)}$. Therefore, the total preprocessing time for all leaves of $T_D$
is $m^5\log m\cdot 2^{O(m^2)}$.
As the decision tree $T_D$ can be built in $O(2^{\poly(m)})$ time, the total preprocessing time is bounded by $O(2^{\poly(m)})$.

\subparagraph{Solving a subproblem $\boldsymbol{T(m^2,m)}$.} Consider a
subproblem $T(m^2,m)$ with a set $P'$ of $m$ points and a set $S'$ of $m^2$
segments. We arbitrarily assign indices to points of $P'$ as
$\{p_1,p_2,\ldots,p_m\}$. By using the decision tree $T_D$, we first find the
leaf $v$ of $T_D$ that corresponds to the refined order type of $P'$, which can be done
using $O(m^2)$ comparisons as the height of $T_D$ is $O(m^2)$.
Let $Q=\{q_1,q_2,\ldots,q_m\}$ be the set of $m$ points associated with $v$. Below
we find for each segment $s\in S'$ its point $p_s$ in $P'$. Let $\ell$ denote
the supporting line of $s$.

We first find two consecutive slopes $k_1$ and $k_2$ in $K_{P'}$ such that the slope
of $\ell$ is in $[k_1,k_2)$. Note that we do
not explicitly have the sorted list of $K_{P'}$, but recall that we have the sorted
list of $K_Q$ stored at $v$. Since $P'$ and $Q$ have the same refined
order type, a slope defined by two points $p_i$ and $p_j$ is in the
$k$-th position of $K_{P'}$ if and only if the slope defined by
two points $q_i$ and $q_j$ is in the $k$-th position of
$K_Q$. Hence, we can search $K_Q$ instead; however, whenever we need to use a slope whose
definition involves a point $q_i\in Q$, we use $p_i$ instead. In this way, we
could find $k_1$ and $k_2$ using $O(\log m)$ comparisons. Further, since we have the
balanced binary search tree storing $K_Q$, we can apply the search lemma of Chan and
Zheng~\cite{ref:ChanHo22} as discussed above to find $k_1$ and $k_2$ using only $O(1-\Delta\Phi)$
comparisons.

Without loss of generality, we assume that $s$ is horizontal.
Let $a$ and $b$ denote the left and right endpoints of $s$, respectively.
Suppose we project all points of $P'$ perpendicularly onto $\ell$. Let
$p_{\pi(1)},p_{\pi(2)},\ldots,p_{\pi(m)}$ be the sorted list following their
projections along $\ell$ from left to right, where $\pi(i)$ is the index of the
$i$-th point in this order. We wish to find the index $i$ such that $a$ is between
$p_{\pi(i-1)}$ and $p_{\pi(i)}$ as well as the index $j$ such that $b$ is between
$p_{\pi(j)}$ and $p_{\pi(j+1)}$.
To this end, we do the following. Since $P'$ and $Q$ have the same refined order type,
if we project all points of $Q$ perpendicularly onto $\ell$, then
$q_{\pi(1)},q_{\pi(2)},\ldots,q_{\pi(m)}$ is the sorted list following their
projections along $\ell$ with the same permutation $\pi(\cdot)$.
Hence, to find the index $i$, we can query $a$ in
the sorted list $q_{\pi(1)},q_{\pi(2)},\ldots,q_{\pi(m)}$, which is maintained
at $v$ due to our preprocessing, but again, whenever we need to use a point $q_{\pi(k)}$, we use $p_{\pi(k)}$ instead.
Using the search lemma of Chan and Zheng as discussed before, we can find $i$ using $O(1-\Delta\Phi)$
comparisons. Similarly, the index $j$ can be found using $O(1-\Delta\Phi)$
comparisons.

Let $P'_{\ell}[i,j]=\{p_{\pi(i)},p_{\pi(i+1)},\ldots,p_{\pi(j)}\}$. By the
definitions of $i$ and $j$, the point
$p_s$ we are looking for is the point of $P'_{\ell}[i,j]$ closest to the line $\ell$. To
find $p_s$, we do the following. Let $\ell'$ be a line parallel to $\ell$
but is below all points of $P'$ and $Q$. Let $P'_{\ell'}[i,j]$ denote the sorted
list of $P'_{\ell}[i,j]$ ordered by their distances from $\ell'$. Then, $p_s$
can be found by binary search on $P'_{\ell'}[i,j]$. Since $P'$ and $Q$ have the
same refined order type, we can instead do binary search on
$Q_{\ell'}[i,j]$, whose order is consistent with that of $Q[i,j]$, which is maintained at $v$ due to the preprocessing. As such we can search $Q[i,j]$, but again whenever the algorithm wants to use a point $q_k\in Q[i,j]$, we will use $p_k$ instead to perform a comparison. Using the search lemma of Chan and Zheng, we can find $p_s$ using $O(1-\Delta\Phi)$ comparisons.

The above shows that $p_s$ can be found using $O(1-\Delta\Phi)$
comparisons. Therefore, doing this for all $O(m^2)$ segments can be done using $O(m^2-\Delta\Phi)$
comparisons.


In summary, with $O(2^{\poly(n)})$ time preprocessing, we can solve each
subproblem $T(n^{2/3},n^{1/3})$ using $O(n^{2/3})$ comparisons without
considering the term $-\Delta\Phi$, whose total sum in the entire algorithm of
recurrence~\eqref{equ:50} is $O(n\log n)$.


\subsection{Wrapping things up}

The above proves Lemma~\ref{lem:10}, and thus $T(n,n)$ in \eqref{equ:50} can be bounded by
$O(n^{4/3})$ after $O(2^{\poly(n)})$ time preprocessing as discussed before. Equivalently,
$T(b,b)$ in \eqref{equ:40} can be bounded by $O(b^{4/3})$ after $O(2^{\poly(b)})$ time
preprocessing. Notice that the preprocessing work is done only once and for all
subproblems $T(b,b)$ in \eqref{equ:40}.
Since $b=(\log\log\log n)^3$, we have $2^{\poly(b)}=O(n)$.
As such, $T(n,n)$ in \eqref{equ:40} solves to $O(n^{4/3})$ and we have the
following conclusion.

\begin{theorem}\label{theo:segment}
Given a set of $n$ points and a set of $n$ segments in the plane, we can find
for each segment its closest point in $O(n^{4/3})$ time.
\end{theorem}

The following solves the asymmetric case of the problem.
\begin{corollary}
Given a set of $n$ points and a set of $m$ segments in the plane, we can find for each segment its closest point in  $O(n^{2/3}m^{2/3}+n\log n+m\log^2 n)$ time.
\end{corollary}
\begin{proof}
Depending on whether $m\geq n$, there are two cases.
\begin{enumerate}
  \item If $m\geq n$, depending on whether $m< n^2$ there are two subcases.

  \begin{enumerate}
    \item If $m<n^2$, then let $r=m/n$, and thus $m/r^2 = n/r$. Applying
	\eqref{equ:20} and solving $T(m/r^2,n/r)$ by Theorem~\ref{theo:segment}
	gives $T(m,n)=O(m\log^2 n+n^{2/3}m^{2/3})$.

    \item If $m\geq n^2$, then applying recurrence \eqref{equ:20} with $r=n$, we
	obtain the following
\begin{align*}
  T(m,n)= O(m\log^2 n)+O(n^2)\cdot T(m/n^2,1).
\end{align*}
For $T(m/n^2,1)$, the problem is to find for each of the $m/n^2$ lines its
closest point among a single point, which can be trivially solved in $O(m/n^2)$
time.
Hence, the above recurrence solves to $O(m\log^2 n)$.
  \end{enumerate}
  Hence in the case where $m\geq n$, we can solve the problem in $O(m\log^2
  n+n^{2/3}m^{2/3})$ time.

  \item If $m< n$, depending on whether $n< m^2$ there are two subcases.
  \begin{enumerate}
	\item If $n<m^2$, then let $r=n/m$, and thus $m/r=n/r^2$. Applying
	\eqref{equ:10} and solving $T(m/r,n/r^2)$ by Theorem~\ref{theo:segment}
	gives $T(m,n)=O(n\log n+n^{2/3}m^{2/3})$.

	\item If $n\geq m^2$, then applying recurrence \eqref{equ:10} with $r=m$, we
	obtain the following
\begin{align*}
  T(m,n)= O(n\log n)+O(m^2)\cdot T(1,n/m^2).
\end{align*}
For $T(1,n/m^2)$, the problem is to find the closes point to a single segment among
$O(n/m^2)$ points, which can be solved in $O(n/m^2)$ time by brute force. As
such, the above recurrence solves to $O(n\log n)$.
  \end{enumerate}
  Hence in the case where $m<n$, we can solve the problem in $O(n\log
  n+n^{2/3}m^{2/3})$ time.
\end{enumerate}

Combining the two cases, the corollary follows.
\end{proof}

\section{A simpler algorithm for the line case}
\label{sec:line}

In this section, we present a simpler solution for the line case, where all
segments of $S$ are lines. The algorithm still runs in $O(n^{4/3})$ time.

In the following, we will present two algorithms, one in the primal plane and the other in the dual plane.
We begin with the first algorithm for the primal plane, which can be viewed as a
simplified version of Bespamyatnikh's algorithm reviewed in Section~\ref{sec:review}.

\subsection{The first algorithm -- in the primal plane}

We again let $n=|P|$ and $m=|S|$.

For a parameter $r$ with $1\leq r\leq \min\{m,\sqrt{n}\}$, we compute a
hierarchical $(1/r)$-cutting $\Xi_0,\Xi_1,\ldots,\Xi_k$ for $S$.
For each cell $\sigma \in \Xi_i$, $0\leq i\leq k$, let $P(\sigma)=P\cap \sigma$; let $S(\sigma)$ denote the subset of the lines of $S$ intersecting $\sigma$. 
We further partition each cell of the last cutting $\Xi_k$ into triangles so
that each triangle contains at most $n/r^2$ points of $P$ and the number of new
triangles in $\Xi_k$ is still bounded by $O(r^2)$. For convenience, we consider
the new triangles as new cells of $\Xi_k$ (we still define $P(\sigma)$ and
$S(\sigma)$ for each new cell $\sigma$ in the same way as above; so we have
$|P(\sigma)|\leq n/r^2$ and $|S(\sigma)|\leq m/r$ for each cell $\sigma\in\Xi_k$).

For each cell $\sigma\in \Xi_k$, we form a subproblem $(S(\sigma),P(\sigma))$ of
size $(m/r,n/r^2)$, i.e., find for each line $\ell$ of $S(\sigma)$ its closest
point in $P(\sigma)$. After the subproblem is solved, to find the closest point
of $\ell$ in $P$, it suffices to find its closest point in $P\setminus
P(\sigma)$. To this end, observe that $P\setminus P(\sigma)$ is exactly the union of
$P(\sigma'')$ for all cells $\sigma''$ such that $\sigma''$ is a child of an
ancestor $\sigma'$ of $\sigma$ and $s\not\in S(\sigma'')$.
As such,
for each of such cells $\sigma''$, we find the closest point of $s$
in $P(\sigma'')$. For this, since $s\not\in S(\sigma'')$, $s$ is outside
$\sigma''$ and thus is outside the convex hull of $P(\sigma'')$.
Hence, it suffices to find the vertex of the convex hull of $P(\sigma')$
closest to $\ell$, which we refer to as an {\em outside-hull line query} and is a much
easier problem than before for the segment case; this
is part of the reason the algorithm is easier in the line case. For answering
the queries, we compute and store the convex hull of $P(\sigma)$ for all cells
$\sigma\in \Xi_i$ for all $i=0,1,\ldots,k$.

For the time analysis, let $T(m,n)$ denote the total time of the
algorithm. Then, solving all subproblems takes $O(r^2)\cdot
T(m/r,n/r^2)$ time.
Constructing the hierarchical cutting as well as computing $S(\sigma)$
for all cells $\sigma$ in all cuttings $\Xi_i$, $0\leq i\leq k$, takes
$O(mr)$ time~\cite{ref:ChazelleCu93}. Computing $P(\sigma)$ for
all cells $\sigma$ can be done in $O(n\log r)$ time.
Computing the convex hulls for $P(\sigma)$ for all cells $\sigma$ in the cutting
can be done in $O(n(\log n/r^2+\log r))$ time in a bottom-up manner.
Indeed, initially, we compute the convex hull for $P(\sigma)$ for every cell
$\sigma\in \Xi_k$ by sorting all points of $P(\sigma)$ first, which takes $O(|P(\sigma)|\log
(n/r^2))$ time since $|P(\sigma)|\leq n/r^2$. After processing all cells of
$\Xi_k$, for each cell $\sigma'$ of $\Xi_{k-1}$, to compute the convex hull of
of $P(\sigma')$, we can sort $P(\sigma')$ by
merging the sorted lists of $P(\sigma)$ for
all children $\sigma$ of $\sigma'$, which have already been computed.
As $\sigma'$ has $O(1)$ children, the
merge can be done in $O(|P(\sigma')|)$ time, and thus computing the convex hull
for $P(\sigma')$ takes only linear time. In this way, the total
preprocessing time for all cells in all cuttings is bounded by
$O(n(\log r+\log (n/r^2)))$ time, which is $O(n\log (n/r))$.
We consider the hierarchical cutting as a tree $T$ such that each node has
$O(1)$ children and each node maintains a convex hull.
We compute a fractional cascading
structure~\cite{ref:ChazelleFr86} on the convex hulls of all nodes of $T$
so that if a tangent to the convex hull at a node $v$
is known, then the tangents of the same slope to the convex hulls of
the children of $v$ can be found in constant time. Constructing the fractional
cascading structure takes time linear in the total size of all convex hulls in
the cutting, which is $O(n\log r)$.

Since $\sum_{i=0}^{k}\sum_{\sigma\in X_i}|S(\sigma)|=O(mr)$, the total number of
outside-hull line queries is $O(mr)$. The total
query time is $O(mr\log n)$, but can be reduced to $O(mr+m\log n)$ using
the fractional cascading structure. Indeed, $S$ has $n$ lines. For
each line $\ell\in S$, for each node $v\in T$ such that $\ell$ crosses the cell
$\sigma$ at $v$, we perform a query on the children $\sigma'$ of $\sigma$ if
$\ell$ does not cross $\sigma'$. Notice that the nodes of $T$ whose cells are
crossed by $\ell$ form a subtree $T_{\ell}$ that contains the root. As such, to
solve all queries for $\ell$, we can start a binary search on the convex hull at
the root using the slope of $\ell$, which takes $O(\log n)$ time, and then solve
each query on other nodes of $T_{\ell}$ in $O(1)$ time each by following the subtree
$T_{\ell}$ in a top-down manner. As each node of $T$ has $O(1)$ children,
answering all queries for $\ell$ takes $O(\log n+|T_{\ell}|)$ time. As such,
solving all queries for all lines $\ell\in S$ takes $O(m\log
n+\sum_{\ell}|T_{\ell}|)$. As $\sum_{\ell\in S}|T_{\ell}|=O(mr)$, the total time for
all queries is $O(m\log n+mr)$.

In summary, we obtain the following
recurrence for any $1\leq r\leq \min\{m,\sqrt{n}\}$
\begin{align}\label{equ:60}
  T(m,n)= O(n\log ({n}/{r}) + m\log n+ mr)+O(r^2)\cdot T(m/r,{n}/{r^2}).
\end{align}
Comparing to \eqref{equ:10} for the segment case, the factor $mr\log n$ is reduced to $m\log n+mr$.

\subsection{The second algorithm -- in the dual plane}

Without loss of generality, we assume that no line of $S$ is vertical.
Let $P^*$ denote the set of all lines dual to the points of $P$ and
$S^*$ the set of all points dual to the lines of $S$.

For each line $\ell\in S$, to find its closest point in $P$, it suffices to find
its closest point among all points of $P$ above $\ell$ and
its closest point among all points of $P$ below $\ell$. In the following, we
only compute for each line $\ell$ of $S$ its closest point among all points of $P$ above $\ell$, since the other case can be handled similarly.
In the dual plane, this is to find for each dual point $\ell^*\in S^*$, the first line of $P^*$
hit by the vertically downward ray $\rho(\ell^*)$ from $\ell^*$.

We compute a $(1/r)$-cutting $\Xi$ for $P^*$ for $1\leq r\leq \min\{n,\sqrt{m}\}$. This time instead of having each cell of $\Xi$ as a triangle, we make each cell of $\Xi$ a trapezoid that is bounded by two vertical edges, an upper edge, and a lower edge. This can be done by slightly changing Chezelle's algorithm~\cite{ref:ChazelleCu93}, i.e., instead of triangulating each cell of a line arrangement, we decompose it into trapezoids (i.e., draw a segment from each vertex of the cell until the cell boundary). Computing such a cutting $\Xi$ can  be done in $O(nr)$ time~\cite{ref:ChazelleCu93}. A property of the cutting produced by Chezelle's algorithm~\cite{ref:ChazelleCu93} is that the upper/lower edge of each trapezoid must lie on a line of $P^*$ unless it is unbounded. For each cell $\sigma$ of $\Xi$, let $P^*(\sigma)$ denote the lines of $P^*$ crossing $\sigma$ and let $S^*(\sigma)=S^*\cap \sigma$. Hence, $|P^*(\sigma)|\leq n/r$. We further cut each cell of $\Xi$ by adding vertical segments so that each new cell contains at most $m/r^2$ points of $S^*$. We still use $\Xi$ to refer to the new cutting. The number of cells of $\Xi$ is still $O(r^2)$.
Computing $P^*(\sigma)$ for all cells $\sigma$ and adding the cutting segments as above to obtain this new cutting $\Xi$ together can be done in $O(nr+m\log r)$ time.

For each cell $\sigma\in \Xi$, we form a subproblem $(S^*(\sigma),P^*(\sigma))$
of size $(m/r^2,n/r)$, i.e., find for each point $\ell^*\in S^*(\sigma)$, the
first line $p^*$ of $P^*(\sigma)$ hit by $\rho(\ell^*)$. A key observation is
that if $p^*$ exists, then it is the ray-shooting answer for $\ell^*$; otherwise,
since $\rho(\ell^*)$ will hit the lower edge of $\sigma$, which lies on a line
$p_1^*\in P^*$, $p_1^*$ is the ray-shooting answer. As such, it
suffices to only solve these subproblems $(S^*(\sigma),P^*(\sigma))$ for all
cells $\sigma\in \Xi$. We thus obtain the following recurrence for any
$1\leq r\leq \min\{n,\sqrt{m}\}$:
\begin{align}\label{equ:70}
  T(m,n)= O(nr + m\log r)+O(r^2)\cdot T({m}/{r^2},{n}/{r}).
\end{align}
Comparing to \eqref{equ:20} for the segment case, the factor $nr\log n$ is reduced to $nr$ and the factor $m\log r\log n$ is reduced to $m\log r$.

\subsection{Combining the two algorithms}

By setting $m=n$ and applying \eqref{equ:60} and \eqref{equ:70} in succession (using the same $r$), we obtain the following recurrence
\begin{align}\label{equ:80}
  T(n,n)= O(n\log n+nr\log r)+O(r^4)\cdot T({n}/{r^3},{n}/{r^3}).
\end{align}
Setting $r=n^{1/3}/\log n$ leads to
\begin{align*}
  T(n,n)= O(n^{4/3}) + O((n/\log^3 n)^{4/3})\cdot T(\log^3 n,\log^3 n).
\end{align*}

By applying the above recurrence three times we can derive the following:
\begin{align}\label{equ:90}
  T(n,n)= O(n^{4/3}) + O((n/b)^{4/3})\cdot T(b,b),
\end{align}
where $b=(\log\log\log n)^3$.

Next we show that after $O(2^{\poly(b)})$ time preprocessing, each $T(b,b)$ can
be solved in $O(b^{4/3})$ time.
For notational convenience, we still use $n$ to represent $b$. Hence, our goal
is to show that after $O(2^{\poly(n)})$ time preprocessing, $T(n,n)$ can
be solved in $O(n^{4/3})$ time. To this end, we show that $T(n,n)$ can
be solved using $O(n^{4/3})$ comparisons, or alternatively, $T(n,n)$ can be
solved by an algebraic decision tree of height $O(n^{4/3})$. The problem now becomes much easier than the segment case.

A close examination of recurrence~\eqref{equ:80} shows that it is the point
location in the above second algorithm that prevents us from obtaining an
$O(n^{4/3})$ time bound for $T(n,n)$; more precisely, each point location introduces an additional logarithmic factor.
To overcome the issue, we can again use the $\Gamma$-algorithm
framework of Chan and Zheng~\cite{ref:ChanHo22}. Indeed, point location is the
main issue Chan and Zheng intended to solve for Hopcroft's problem. For this,
Chan and Zheng proposed the basic search lemma. We can follow the similar idea as
theirs (see Lemma~4.2~\cite{ref:ChanHo22}).

We modify the second algorithm with the following change.  To find the cell of $\Xi$ containing
each point of $S^*$, we apply Chan and Zheng's basic search lemma on the $O(r^2)$ cells of
$\Xi$, which can be done using $O(1-r^2\Delta\Phi)$ comparisons (instead of
$O(\log r)$). Excluding the $O(-r^2\Delta\Phi)$ terms, we obtain a new
recurrence for any $1\leq r\leq \min\{n,\sqrt{m}\}$:
\begin{align}\label{equ:100}
  T(m,n)= O(nr + m)+O(r^2)\cdot T(m/r^2,n/r).
\end{align}

Using the same $r$, we stop the recursion until $n=\Theta(r)$, which is the
base case. In the base case we have $T(m,n)=O(m+r^2)$ (again excluding the term
$O(-r^2\Delta\Phi)$) by simply constructing the
vertical decomposition of the $\Theta(r)$ dual lines in $O(r^2)$ time and then apply the basic
search lemma to find the cell of the decomposition containing each point.
In this way, the recurrence~\eqref{equ:100} solves to $T(m,n)=(n^2+m)\cdot 2^{O(\log_rn)}$.
By setting $r=n^{\epsilon/2}$, we obtain the following bound
on the number of comparisons excluding the term
$O(-n^{\epsilon}\Delta\Phi)$.
\begin{align}\label{equ:110}
T(m,n)=O(n^2+m).
\end{align}

Now we apply recurrence \eqref{equ:60} with $m=n$ and $r=n^{1/3}$ and obtain the
following
\begin{align}\label{equ:120}
T(n,n)=O(n^{4/3}) + O(n^{2/3})\cdot T(n^{2/3},n^{1/3}).
\end{align}
Applying \eqref{equ:110} for $T(n^{2/3},n^{1/3})$ gives $T(n,n)=O(n^{4/3})$ with the excluded terms sum
to $O(n^{\epsilon}\cdot n\log n)$. As such, by setting $\epsilon$ to a small value (e.g., $\epsilon=1/4$), we conclude that $T(n,n)$ can be solved using
$O(n^{4/3})$ comparisons, or alternatively, we have an algebraic decision tree of height $O(n^{4/3})$ that can solve $T(n,n)$.

Now back to the recurrence~\eqref{equ:90}, i.e., our original problem, we apply the above decision tree algorithm on $T(b,b)$. If we build the decision tree beforehand, which can be done in $2^{\poly(b)}$ time, then we can bound the time for $T(b,b)$ by $O(b^{4/3})$. Note that we only build one decision tree and use it to solve all subproblems $T(b,b)$. As $b=O((\log\log\log n)^3)$, we have $2^{\poly(b)}=O(n)$. Hence, the total time of the algorithm is bounded by $O(n^{4/3})$.

\begin{theorem}\label{theo:line}
Given a set of $n$ points and a set of $n$ lines in the plane, we can find
for each line its closest point in $O(n^{4/3})$ time.
\end{theorem}

The following solves the asymmetric case of the problem.
\begin{corollary}
Given a set of $n$ points and a set of $m$ lines in the plane, we can find for each line its closest point in $O(n^{2/3}m^{2/3}+(n+m)\log n)$ time.
\end{corollary}
\begin{proof}
Depending on whether $m\geq n$, there are two cases.
\begin{enumerate}
  \item If $m\geq n$, depending on whether $m< n^2$, there are two subcases.

  \begin{enumerate}
    \item If $m<n^2$, then let $r=m/n$, and thus $m/r^2 = n/r$. Applying \eqref{equ:70} and solving $T(m/r^2,n/r)$ by Theorem~\ref{theo:line} give $T(m,n)=O(m\log n+n^{2/3}m^{2/3})$.

    \item If $m\geq n^2$, then we solve the problem in the dual plane. We first construct the vertical decomposition $D$ of the arrangement of the dual lines of the points of $S$ in $O(n^2)$ time and then build a point location data structure on the decomposition in $O(n^2)$ time~\cite{ref:KirkpatrickOp83,ref:EdelsbrunnerOp86}. Next, for each dual point of each line of $S$, we find the cell of $D$ that contains the point in $O(\log n)$ time using the point location data structure. This takes $O(n^2+m\log n)$ time in total, which is $O(m\log n)$ as $m\geq n^2$.
  \end{enumerate}
   Hence in the case where $m\geq n$, we can solve the problem in $O(m\log n+n^{2/3}m^{2/3})$ time.

  \item If $m< n$, depending on whether $n< m^2$, there are two subcases.
  \begin{enumerate}
	\item If $n<m^2$, then let $r=n/m$, and thus $m/r=n/r^2$. Applying
	\eqref{equ:60} and solving $T(m/r,n/r^2)$ by Theorem~\ref{theo:line} give $T(m,n)=O(n\log m+n^{2/3}m^{2/3})$.

	\item If $n\geq m^2$, then applying recurrence \eqref{equ:60} with $r=m$, we
	obtain the following
\begin{align*}
  T(m,n)= O(n\log n)+O(m^2)\cdot T(1,n/m^2).
\end{align*}
For $T(1,n/m^2)$, the problem is to find the closest point to a single line among
$O(n/m^2)$ points, which can be solved in $O(n/m^2)$ time by brute force.
Hence, the above recurrence solves to $O(n\log n)$.
  \end{enumerate}
  Hence in the case where $m<n$, we can solve the problem in $O(n\log
  n+n^{2/3}m^{2/3})$ time.
\end{enumerate}

Combining the two cases, the corollary follows.
\end{proof}

\section{The line query problem}
\label{sec:linequery}

In this section, we discuss the query problem for the line case. Let $P$ be a set of $n$ points in the plane. We wish to build a data structure so that the point of $P$ closest to a query line $\ell$ can be computed efficiently. The main idea is to adapt the simplex range searching data structures~\cite{ref:ChanOp12,ref:MatousekEf92,ref:MatousekRa93} (which works in any fixed dimensional space; but for our purpose it suffices to only consider half-plane range counting queries in the plane).

The rest of this section is organized as follows. After giving an overview of our approach, we present a randomized result based on Chan's partition tree~\cite{ref:ChanOp12} in Section~\ref{sec:chan}. In the subsequent two subsections we present two deterministic results, one based on Matou\v{s}ek's partition
tree~\cite{ref:MatousekEf92} and the other based on Matou\v{s}ek's hierarchical cuttings~\cite{ref:MatousekRa93}. Finally in Section~\ref{sec:tradeoff} we derive trade-offs between preprocessing and query time.

\subparagraph{Overview.}
Each of these half-plane range counting query data structures~\cite{ref:ChanOp12,ref:MatousekEf92,ref:MatousekRa93} defines canonical subsets of $P$ and usually only maintains the cardinalities of them. To solve our problem, roughly speaking, the change is that we compute and maintain the convex hulls of these canonical subsets, which increases the space by a factor proportional to the height of the underlying trees (which is $O(\log n)$ for the data structures in~\cite{ref:ChanOp12,ref:MatousekRa93} and is $O(\log\log n)$ for the one in~\cite{ref:MatousekEf92}). To answer a query, we follow the similar algorithms as half-plane range counting queries on these data structures. The difference is that for certain canonical subsets, we do binary search on their convex hulls to find their closest vertices to the query line, which does not intersect these convex hulls (in the half-plane range counting query algorithms only the cardinalities of these canonical subsets are added to a total count). This increases the query time by a logarithmic factor comparing to the original half-plane range counting query algorithms. We manage to reduce the additional logarithmic factor using fractional cascading~\cite{ref:ChazelleFr86} on the data structures of~\cite{ref:ChanOp12,ref:MatousekRa93} because each node in the underlying trees of these data structures has $O(1)$ children.
Some extra efforts are also needed to achieve the claimed performance.
Finally, the trade-off is obtained by combining these results with cuttings in the dual space.

\subsection{A randomized result based on Chan's partition tree~\cite{ref:ChanOp12}}
\label{sec:chan}

We first review Chan's partition tree~\cite{ref:ChanOp12}.
Chan's partition tree $T$ for the point set $P$ is a tree structure by recursively subdividing the plane into triangles. Each node $v$ of $T$ is associated with a triangle $\triangle(v)$, which is the entire plane if $v$ is the root. If $v$ is an internal node, it has $O(1)$ children, whose associated triangles form a disjoint partition of $\triangle(v)$. Let $P(v)=P\cap \triangle(v)$, i.e., the subset of points of $P$ in $\triangle(v)$. For each internal node $v$, the cardinality $|P(v)|$ is stored at $v$. If $v$ is a leaf, then $|P(v)|=O(1)$ and $P(v)$ is explicitly stored at $v$. The height of $T$ is $O(\log n)$ and the space of $T$ is $O(n)$. Let $\alpha(T)$ denote the maximum number of triangles $\triangle(v)$ among all nodes $v$ of $T$  crossed by any line in the plane. Given $P$, Chan's randomized algorithm can compute $T$ in $O(n\log n)$ expected time such that $\alpha(T)=O(\sqrt{n})$ holds with high probability.

To solve our problem, we modify the tree $T$ as follows. For each node $v$, we compute the convex hull $\ch(v)$ of $P(v)$ and store $\ch(v)$ at $v$. This increases the space to $O(n\log n)$, but the preprocessing time is still bounded by $O(n\log n)$.

Given a query line $\ell$, our goal is to compute the point of $P$ closest to $\ell$. We only discuss how to find the closest point of $\ell$ among all points of $P$ below $\ell$ since the other case is similar. Starting from the root of $T$, consider a node $v$. We assume that $\ell$ crosses $\triangle(v)$, which is true initially when $v$ is the root. For each child $u$ of $v$, we do the following. If $\ell$ crosses $\triangle(u)$, then we proceed on $u$ recursively. Otherwise, if $\triangle(u)$ is below $\ell$, we do binary search on the convex hull $\ch(u)$ to find in $O(\log n)$ time the closest point to $\ell$ among the vertices of $\ch(u)$ and keep the point as a candidate. Since each internal node of $T$ has $O(1)$ children, the algorithm eventually finds $O(\alpha(T))$ candidate points and among them we finally return the one closest to $\ell$ as our solution. The total time of the algorithm is $O(\alpha(T)\cdot \log n)$.

To further reduce the query time, we observe that all nodes $v$ whose triangles
$\triangle(v)$ are crossed by $\ell$ form a subtree $T_{\ell}$ of $T$ containing the root.
This is because if the triangle $\triangle(v)$ of a node $v$ is crossed by
$\ell$, then the triangle $\triangle(u)$ is also crossed by $\ell$
for any ancestor $u$ of $v$. In light of the observation, we can 	
further reduce the query algorithm time to $O(\alpha(T)+\log n)$ by
constructing a fractional cascading
structure~\cite{ref:ChazelleFr86} on the convex hulls of all nodes of $T$
so that if a tangent to the convex hull at a node $v$
is known, then the tangents of the same slope to the convex hulls of
the children of $v$ can be found in constant time. The total time for
constructing the fractional cascading structure is linear in the total size of
all convex hulls, which is $O(n\log n)$. With the fractional cascading structure, we only
need to perform binary search on the convex hull at the root and then spend only $O(1)$ time on each node of $T_{\ell}$ and each of their children. As such, the query time becomes $O(\alpha(T)+ \log
n)$, which is bounded by $O(\sqrt{n})$ with high probability.


The following lemma summarizes the result.
\begin{lemma}\label{lem:linechan}
Given a set $P$ of $n$ points in the plane, we can build a data structure of $O(n\log n)$ space in $O(n\log n)$ expected time such that for any query line its closest point in $P$ can be computed in $O(\sqrt{n})$ time with high probability.
\end{lemma}

\subsection{A deterministic result based on Matou\v{s}ek partition
tree~\cite{ref:MatousekEf92}}
\label{sec:linem1}

We now present a deterministic result based on Matou\v{s}ek partition
tree~\cite{ref:MatousekEf92}. We first briefly review the partition tree in the plane for half-plane range counting queries.

A {\em simplicial partion} of size $m$ for the point set $P$ is a collection
$\Pi=\{(P_1,\triangle_1),\ldots,(P_m,\triangle_m)\}$ with the
following properties: (1) The subsets $P_i$'s form a disjoint
partition of $P$; (2) each cell $\triangle_i$ is an open triangle containing $P_i$; (3) $\max_{1\leq i\leq m}|P_i|\leq 2\cdot \min_{1\leq i\leq m}|P_i|$; (4) the cells may overlap and any cell $\triangle_i$ may contain points in $P\setminus P_i$. We define the {\em crossing number} of $\Pi$ as the largest number of cells crossed by any line in the plane.

\begin{lemma}\label{lem:simpartition}{\em (Matou\v{s}ek~\cite{ref:MatousekEf92})}
For any $s$ with $2\leq s\leq |P|$, there exists a simplicial partition $\Pi$ for $P$, whose subsets $P_i$'s satisfy $s\leq |P_i|<2s$, and whose crossing number is $O(\sqrt{r})$, where $r=|P|/s$.
\end{lemma}

\begin{lemma}\label{lem:simpartitionalgo}{\em (Matou\v{s}ek~\cite{ref:MatousekEf92})}
For any fixed $\delta>0$ and $s\geq n^{\delta}$, a simplicical partition whose subsets $|P_i|$ satisfy $s\leq |P_i|<2s$ and whose crossing number is $O(\sqrt{r})$ can be constructed in $O(n\log r)$ time, where $r=|P|/s$.
\end{lemma}

Matou\v{s}ek's algorithm~\cite{ref:MatousekEf92}
builds a half-plane range counting data structure in $O(n)$ space and $O(n\log
n)$ time as follows. The data structure is a {\em partition tree} $T$, which is
built by applying Lemma~\ref{lem:simpartitionalgo} recursively to partition $P$
into subsets of constant sizes, which form the leaves of $T$. Each internal node
$v$ of $T$ corresponds to a subset $P(v)$ of $P$ as well as a simplicial
partition $\Pi(v)$ of $P(v)$, which form the children of $v$. At each child $u$
of $v$, the cell $\triangle(u)$ of $\Pi(v)$ containing $P(u)$ and the
cardinality $|P(u)|$ are stored at $u$. In particular, $\triangle(u)$ is the entire plane if
$u$ is the root. The simplicial partition $\Pi(v)$ is computed by
Lemma~\ref{lem:simpartitionalgo} with $s=\sqrt{|P(v)|}$. As such, the height of
$T$ is $O(\log \log n)$. The time for constructing $T$ is $O(n\log n)$ because the size of $|P(v)|$ is geometrically decreasing in a top-down manner.

Given a query half-plane $h$ bounded by a line $\ell$, the algorithm for
computing the number of points of $P$ in $h$ works as follows. Starting from the
root of $T$, consider a node $v$. We assume that $\ell$ crosses $\triangle(v)$, which is
true initially when $v$ is the root. We check every child $u$ of $v$. If
$\triangle(u)$ is crossed by $\ell$, we proceed on $u$. Otherwise, if
$\triangle(u)$ is inside $h$, we add $|P(u)|$ to a total count. It can be
shown that the query time is bounded by $O(\sqrt{n}\log^{O(1)}
n)$~\cite{ref:MatousekEf92}; alternatively, the number of nodes visited
by the algorithm is bounded by $O(\sqrt{n}\log^{O(1)}
n)$~\cite{ref:MatousekEf92}.

We now modify the data structure for our problem. For each node $u$ of $T$, we
compute and store the convex hull $\ch(u)$ of $P(u)$ at each node $u$. As the
height of $T$ is $O(\log\log n)$, the total space increases to $O(n\log\log n)$.
If we pre-sort all points of $P$, for each node $u$, we can construct the convex
hull in $O(|P(u)|)$ time and thus the total preprocessing is still bounded by
$O(n\log n)$.

Given a query line $\ell$, our goal is to compute the point of $P$ closest to
$\ell$. We only discuss how to find the closest point of $\ell$ among all points
of $P$ below $\ell$ since the other case is similar. Starting from the root of
$T$, consider a node $v$. We assume that $\ell$ crosses $\triangle(v)$, which is
true initially when $v$ is the root. For each child $u$ of $v$, we do the
following. If $\ell$ crosses $\triangle(u)$, we proceed on $u$ recursively.
Otherwise, if $\triangle(u)$ is below $\ell$, we do binary search on the
convex hull $\ch(v)$
to find in $O(\log n)$ time the closest point to $\ell$ among the vertices of
$\ch(v)$ and keep the point as a candidate. Finally, among all candidate
points we return the one closest to $\ell$. Notice that the number of nodes
visited by the query algorithm is the same as that in the half-plane range
counting query algorithm. Thus, the total time of the algorithm is still
$O(\sqrt{n}\log^{O(1)} n)$ (our algorithm spend $O(\log n)$ additional time on each
visited node).


The following lemma summarizes the result.
\begin{lemma}\label{lem:linem1}
Given a set $P$ of $n$ points in the plane, we can build a data structure of
$O(n\log\log n)$ space in $O(n\log n)$ time such that for any query
line its closest point in $P$ can be computed in
$O(\sqrt{n}\log^{O(1)}n)$ time.
\end{lemma}

\subsection{A deterministic result based on Matou\v{s}ek's hierarchical
cuttings~\cite{ref:MatousekRa93}}
\label{sec:linem2}

We present another deterministic result based on Matou\v{s}ek's another simplex
range searching data structure~\cite{ref:MatousekRa93}, which makes uses of
Chazelle's result on hierarchical cuttings~\cite{ref:ChazelleCu93}. We first
briefly review Matou\v{s}ek's data structure~\cite{ref:MatousekRa93}. The data structure works for simplex range counting queries in any fixed dimensional space. Again for simplicity and for our purpose, we only discuss it for half-plane range counting queries in the plane.

The algorithm first constructs a data structure for a subset $P'$ of at least half
points of $P$. To build a data structure for the whole $P$, the same
construction is performed for $P$, then for $P\setminus P'$, etc., and
thus $O(\log n)$ data structures with geometrically
decreasing sizes will be obtained. Because both the preprocessing time and space of
the data structure for $P'$ are $\Omega(n)$,
constructing all data structures for $P$ takes asymptotically the same time and space
as those for $P'$ only. To answer a half-plane range counting query on $P$, each of these data
structures will be called. Since the query time for $P'$ is
$\Omega(\sqrt{n})$, the total query time for $P$ is asymptotically the
same as that for $P'$. Below we describe the data structure for $P'$.

The data structure has a set of (not necessarily disjoint) cells that are triangles,
$\Psi_0=\{\triangle_1,\ldots,\triangle_t\}$ with $t=\sqrt{n}\log n$.
For each $1\leq i\leq t$, a subset $P_i\subseteq P$ of
$\frac{n}{2t}$ points contained in $\triangle_i$ will be computed.
The subsets $P_i$'s form a disjoint partition of $P'$.
A rooted tree $T_i$ is constructed for each $1\leq i\leq t$ such that each node of $T_i$
corresponds to a cell, which is a triangle, with $\triangle_i$ as the root. Each internal node of $T_i$
has $O(1)$ children whose cells are
interior-disjoint and together cover their parent cell. For each
cell $\triangle(v)$ of a node $v$ of $T_i$, let $P(v)=P_i\cap \triangle(v)$.
If $\triangle(v)$ is a leaf, then the points of $P(v)$ are
explicitly stored at $v$; otherwise only $\triangle(v)$ and
$|P(v)|$ are stored at $v$. Each point
of $P_i$ is stored in exactly one leaf cell of $T_i$.
The depth of each $T_i$ is $\kappa=O(\log n)$. Hence, the data structure is a forest of $t$
trees. Let $\Psi_j$ denote the set of all cells of all trees
$T_i$'s that lie at distance $j$ from the root (note that $\Psi_0$ is consistent with this definition). For any line $\ell$ in the
plane, let $K_j(\ell)$ be
the set of cells of $\Psi_j$ crossed by $\ell$; let $L_j(\ell)$ be the set of
leaf cells of $K_j(\ell)$. Define $K(\ell)=\bigcup_{j=0}^{\kappa} K_j(\ell)$ and
$L(\ell)=\bigcup_{j=0}^{\kappa} L_j(\ell)$.
Matou\v{s}ek~\cite{ref:MatousekRa93} proved the following lemma.

\begin{lemma}\label{lem:bound}{\em (Matou\v{s}ek~\cite{ref:MatousekRa93})}
\begin{enumerate}
\item $\sum_{j=0}^{\kappa}|\Psi_j|=O(n)$.
\item
For any line $\ell$ in the plane, $|K(\ell)|=O(\sqrt{n})$.
\item
For any line $\ell$ in the plane, $\sum_{v\in L(\ell)}|P(v)|=O(\sqrt{n})$.
\end{enumerate}
\end{lemma}

To construct the data structure described above, Matou\v{s}ek~\cite{ref:MatousekRa93} gave an algorithm whose runtime is polynomial in $n$, and the space is $O(n)$ due to Lemma~\ref{lem:bound}(1).

We next discuss our new algorithm for our problem. Using Matou\v{s}ek's algorithm~\cite{ref:MatousekRa93} we compute $\Psi_0=\{\triangle_1,\ldots,\triangle_t\}$ with $t=\sqrt{n}\log n$ as well as $T_i$, $\triangle_i$, and $P_i$ for all $i=1,2,\ldots,t$ in the same way as above. In addition, for each node $v$ of each tree $T_i$, we compute the convex hull of $P(v)$ and store it at $v$; let $\ch(v)$ denote the convex hull. As the height of each $T_i$ is $O(\log n)$, the total space increases to $O(n\log n)$.
Further, for each $T_i$, we construct a fractional cascading structure~\cite{ref:ChazelleFr86} on the convex hulls of all nodes of $T_i$ so that if a tangent to the convex hull at a node $v$ is known, then the tangents of the same slope to the convex hulls of the children of $v$ can be found in constant time. The total space is still $O(n\log n)$.

Given a query line $\ell$, our goal is to compute the point of $P$ closest to $\ell$. We only discuss how to find the closest point of $\ell$ among all points of $P$ below $\ell$ since the other case is similar.
We show how to compute the closest point of $\ell$ among all points of $P'$ below $\ell$ and then apply the same algorithm on other $O(\log n)$ subsets, so that we obtain a total of $O(\log n)$ candidate closest points. Finally, among all candidate points we return the one closest to $\ell$.

To find the closest point to $\ell$ among all points of $P'$ below $\ell$, our algorithm consists of the following four steps.
\begin{enumerate}
  \item
   Compute the point closest to $\ell$ among the points inside the cells $\triangle_i$ of $\Psi_0$ that are below $\ell$, and add the point to a set $S$ as a candidate point.
  \item
  Find the subset $\Sigma(\ell)$ of cells $\triangle_i$ of $\Psi_0$ that are crossed by $\ell$.
  \item
  For each $\triangle_i\in \Sigma(\ell)$, we find the closest point to $\ell$ among the points of $P_i$ below $\ell$ and add the point to $S$.
  \item
  Among all points of $S$, return the point closest to $\ell$ by checking every point of $S$.
\end{enumerate}
In what follows, we discuss the details for implementing the first three steps; the fourth step is trivial. We will show that after $O(\poly(n))$ time and $O(n\log n)$ space preprocessing, these steps can be implemented in $O(\sqrt{n})$ time for any $\ell$. The preprocessing time will be further reduced later.

\subparagraph{The first step.}
For the first step, we have the following lemma (recall that $t$ is the number of cells of $\Psi_0$ and $t=\sqrt{n}\log n$).

\begin{lemma}\label{lem:firststep}
With $O(t(\log t)^{O(1)}+n\log n)$ time and $O(t(\log t)^{O(1)}+n\log\log t)$ space preprocessing, the first step can be executed in $O(\sqrt{t}\cdot \exp(c\cdot \sqrt{\log t})\cdot \log n)$ time, for a constant $c$.
\end{lemma}
\begin{proof}
Let $h$ be the half-plane below $\ell$. Our goal is to compute the point closest to $\ell$ among the points inside the cells $\triangle_i$ of $\Psi_0$ that are completely contained in $h$. Note that $\triangle_i$ is in $h$ if and only if all three vertices of $\triangle_i$ are in $h$. We consider the three vertices of each cell $\triangle_i$ as a $3$-tuple and build a 3-level data structure for all cells of $\Psi_0$ by modifying a multi-level data structure of Lemma~6.2 in \cite{ref:MatousekEf92}.
Let $A$ be the set of all $3$-tuples for all $\triangle_i\in \Psi_0$. Hence, $|A|=t$.

We proceed by induction on $k$ with $k=1,2,3$, i.e., solving the $k$-tuple problem by constructing a data structure $D_k(A)$. For $k=1$ (i.e., $A$ is a set of $t$ points), we construct a half-plane range counting data structure of Matou\v{s}ek~\cite{ref:MatousekEf92} as reviewed in Section~\ref{sec:linem1} on $A$ with the following change: for each node $v$ of the partition tree, instead of storing $|A(v)|$ at $v$, we store the convex hull of the points of the union of the subsets $P_i$ for all cells $\triangle_i\in \Psi_0$ that have a vertex in $A$. In this way, $D_1(A)$ can be constructed in $O(t\log t + n\log n)$ time and $O(t+n\log\log t)$ space, following the analysis in Section~\ref{sec:linem1}.

For $k>1$, let $F$ be the set of first elements of all $k$-tuples of $A$.
To construct $D_k(A)$, we build a partition tree as before on $F$, by setting $r=m/s=\exp(\sqrt{\log m})$ in a node $v$ whose subset $F(v)$ has $m$ points. For each subset $F_i$ of the simplicial partition $\Pi(v)=\{(F_1,\triangle_1),(F_2,\triangle_2),\ldots\}$ for $F(v)$, we let $A_i\subseteq A$  be the set of $k$-tuples whose first elements are in $F_i$, and let $A_i'$ be the set of $(k-1)$-tuples arising by removing the first element from the $k$-tuples of $A_i$. We compute the data structure $D_{k-1}(A_i')$ and store it at the node $v$.

Given a query line $\ell$, let $h$ be the half-plane below $\ell$. Start from the root of the partition tree, for each node $v$, we find the cells of the simplicial partition $\Pi(v)$ contained in $h$, and for each such cell $\triangle_i$, we use the data structure $D_{k-1}(A_i')$ to find the $k$-tuples of $A_i$ contained in $h$; in particular, when $k=1$, we do binary search on the convex hull stored at $\triangle_i$ to find the vertex closest to $\ell$ as a candidate closest point. We also find the cells of $\Pi(v)$ crossed by $\ell$, and visit the corresponding subtrees of $v$ recursively. Finally, among all candidate closest points we return the one closest to $\ell$ as the answer.

For the preprocessing time and space, we separate the algorithm into two parts: the part on processing the $3t$ vertices of all cells of $\Psi_0$ and the part on the points of $\cup_{i=1}^tP_i$. For the first part, we can follow exactly the same time analysis as Lemma~6.2~\cite{ref:MatousekEf92} and obtain that the time and space are bounded by $O(t(\log t)^{O(1)})$. For the second part, it is only processed in the lowest level of the partition tree; as discussed above, the total time is $O(n\log n)$ and the space is $O(n\log\log t)$. For the query time, following exactly the same analysis as Lemma~6.2~\cite{ref:MatousekEf92}, we can obtain that the number of nodes of the partition tree accessed by the query algorithm is $O(\sqrt{t}\cdot \exp(c\cdot \sqrt{\log t}))$. Since we need to do binary search on the convex hulls in the lowest level partition tree, the total time is bounded by $O(\sqrt{t}\cdot \exp(c\cdot \sqrt{\log t})\cdot \log n)$.
\end{proof}

Note that $\exp(c\sqrt{\log t})=O(t^{\delta})$ for any small $\delta>0$. Since $t=\sqrt{n}\log n$, the preprocessing time of the above lemma is bounded by $O(n\log n)$, the space is $O(n\log\log n)$, and the query time is bounded by $O(\sqrt{n})$.
In summary, with additional $O(n\log n)$ time and $O(n\log\log n)$ space preprocessing, the first step of the query algorithm can be executed in $O(\sqrt{n})$ time.

\subparagraph{The second step.}
For the second step, we have the following lemma.

\begin{lemma}
With $O(t(\log t)^{O(1)})$ time and space preprocessing, the second
step can be executed in $O(\sqrt{t}\cdot(\log n)^{O(1)}+k)$ time, where $k$ is the output size.
\end{lemma}
\begin{proof}
The second step of the query algorithm is to report all cells of $\Psi_0$ that are crossed by the query line $\ell$.
Let $E$ be the set of the edges of all cells of $\Psi_0$. Hence, $|E|\leq 3t$. Given a query line $\ell$, it suffices to report all edges of $E$ crossed by $\ell$. For this problem, we can use the result in
Lemma~6.3~\cite{ref:MatousekRa93} and the performance is as stated in the lemma.
\end{proof}

Since $t=\sqrt{n}\log n$, the preprocessing time and space of the above lemma are bounded by $O(n)$ and the query time is bounded by $O(\sqrt{n})$, for $k=O(\sqrt{n})$ by Lemma~\ref{lem:bound}(2).

\subparagraph{The third step.}
For the third step, for each cell $\triangle_i\in \Sigma(\ell)$, using the tree $T_i$, we find the closest point of $\ell$ in $P_i\cap h$ as follows, where $h$ is the half-plane below $\ell$. Starting from the root of $T_i$, consider a node $v$. We assume that $\ell$ crosses the cell $\triangle(v)$ at $v$, which is true initially when $v$ is the root. If $v$ is a leaf, then we check points of $P(v)$ one by one and find the point of $P(v)\cap h$ closest to $\ell$ and add it to $P_i(\ell)$ as a candidate point. Otherwise, for each child $u$ of $v$, if $\ell$ crosses $\triangle(u)$, then we proceed on $u$ recursively; otherwise, if $\triangle(u)$ is in $h$, then we find the vertex of $\ch(u)$ closest to $\ell$ and add it to $P_i(\ell)$. Finally, among all points of $P_i(\ell)$, we return the one closest to $\ell$.

For the time analysis, since each internal node of $T_i$ has $O(1)$ children, $|P_i(\ell)|$ is on the order of the number of nodes of $T_i$ whose cells are crossed by $\ell$. As such, by Lemma~\ref{lem:bound}(2), we have $\sum_{\triangle_i\in \Sigma(\ell)}|P_i(\ell)|=O(\sqrt{n})$. Also, by Lemma~\ref{lem:bound}(3), the total time of the algorithm spends on all leaves whose cells are crossed by $\ell$ for all trees $T_i$ with $\triangle_i\in \Sigma(\ell)$ is $O(\sqrt{n})$. It remains to bound the time on finding the closest vertex on the convex hulls $\ch(v)$ for those vertices $v$ whose cells are in $h$ and whose parent cells are crossed by $\ell$. Let $A_i$ be the set of nodes of $T_i$ whose cells are crossed by $\ell$. The nodes whose convex hulls we need to search are children of the nodes of $A_i$ that are not in $A_i$.
Observe that nodes of $A_i$ form a subtree including the root. Using the fractional cascading data structure built on the convex hulls of the nodes of $T_i$, we only need to perform binary search on the convex hull at the root and then spend only $O(1)$ time on each node of $A_i$ (and each of their children). Since each node of $T_i$ has $O(1)$ children, the total time is $O(\log n+ |A_i|)$. By Lemma~\ref{lem:bound}(2), we obtain that $\sum_{\triangle_i\in \Xi}|A_i|=O(\sqrt{n})$. Hence, the total time of the third step is bounded by $O(|\Sigma(\ell)|\cdot \log n+\sqrt{n})$. To bound $|\Sigma(\ell)|$, since each cell of $\Sigma(\ell)$ is crossed by $\ell$, we have $|\Sigma(\ell)|=O(\sqrt{n})$ by Lemma~\ref{lem:bound}(2), which would lead to an overall $O(\sqrt{n}\log n)$ time bound for the query algorithm. Using some results from~\cite{ref:MatousekRa93}, we can obtain the following lemma and consequently bound the query time by $O(\sqrt{n})$.

\begin{lemma}\label{lem:Sigma}
$|\Sigma(\ell)|=O(\sqrt{n}/\log n)$.
\end{lemma}
\begin{proof}
Following the definition in~\cite{ref:MatousekRa93}, a set $G$ of three lines is called a {\em guarding set} for a line $\ell$ with respect to the point set $P$ and a parameter $r\leq n$ if the zone of $\ell$ in the arrangement of $G$ contains less than $n/\sqrt{r}$ points of $P$. Matou\v{s}ek's algorithm~\cite{ref:MatousekRa93} for constructing the half-plane range counting data structure described above for the subset $P'$ starts with computing a collection $\Gamma$ of guarding sets with $r=n$.
Let $H$ be the set of all lines of all guarding sets of $\Gamma$ ($H$ is also called the {\em test set} in the literature, e.g.,~\cite{ref:ChanOp12,ref:MatousekEf92}).\footnote{Computing the guarding sets is done by computing a $(1/t')$-cutting $\Xi$ for the dual lines of the points of $P$ such that $\Xi$ has at most $r$ vertices in total, where $t'$ is chosen so that $t'=\Theta(\sqrt{r})$. $H$ is just the set of lines in the primal plane dual to the vertices of $\Xi$. For any line $\ell$ in the primal plane, its dual point is contained in a cell $\triangle$ (which is a triangle) of $\Xi$; the dual lines of the three vertices of $\triangle$ form the guarding set of $\ell$.}

For any line $\ell\in H$, it is proved in~\cite{ref:MatousekRa93} (i.e., Lemma~4.2) that $|K_0(\ell)|=O(\sqrt{n}/4^{\kappa})$ with $\kappa=\Theta(\log n)$. More precisely, by examining Matou\v{s}ek's algorithm~\cite{ref:MatousekRa93} for constructing $\Psi_0$, we have $\kappa=c\cdot \log_{\rho}(n/\sqrt{t})$ for a constance $c>0$ and another constant $\rho>4$ ($\rho$ is the constant as defined in the hierarchical cutting in Section~\ref{sec:pre}). As $t=\sqrt{n}\log n$, one can verify that for any $c>0$ and any $\rho>4$, there always exists sufficient large $n$ such that $4^{\kappa}>\log n$.
Therefore, it holds that $\sqrt{n}/4^{\kappa}=O(\sqrt{n}/\log n)$ and thus $|K_0(\ell)|=O(\sqrt{n}/\log n)$. Note that $\Sigma(\ell)$ is $K_0(\ell)$. Hence, the lemma holds for any line $\ell\in H$.

Consider a line $\ell\not\in H$. Let $G$ be the guarding set of $\ell$. We partition $\Sigma(\ell)$ into two subsets: let $\Sigma_1(\ell)$ be the subset of cells of $\Sigma_1(\ell)$ crossed by the lines of $G$ and let $\Sigma_2(\ell)=\Sigma(\ell)\setminus\Sigma_1(\ell)$. By the above analysis, since $|G|=3$ and $G\subseteq H$, we have $|\Sigma_1(\ell)|=O(\sqrt{n}/\log n)$.
On the other hand, each cell of $\Sigma_2(\ell)$ must be contained in the zone $Z(\ell)$ of $\ell$ in the arrangement of $G$. According to the property of the guarding set, $Z(\ell)$ contains at most $n/\sqrt{r}=\sqrt{n}$ points of $P$. Recall that $|P_i|=n/(2t)$ for each cell $\triangle_i$ of $\Psi_0$, $1\leq i\leq t$. Hence, the number of cells of $\Psi_0$ in $Z(\ell)$ is at most $2t\sqrt{n}/n=2\log n$. Therefore, we obtain $|\Sigma_2(\ell)|\leq 2\log n$. As such, $|\Sigma(\ell)|=|\Sigma_1(\ell)|+|\Sigma_2(\ell)|=O(\sqrt{n}/\log n)$. The lemma thus follows.
\end{proof}


\subparagraph{Summary.}
The above shows that the first three steps of the query algorithm can be
executed in $O(\sqrt{n})$ time, with additional $O(n\log n)$ time and
$O(n\log\log n)$ space preprocessing.
For the fourth step, its time is $O(|S|)$. Notice that $|S|\leq
|\Sigma(\ell)|+1$, which is $O(\sqrt{n}/\log n)$ by Lemma~\ref{lem:Sigma}.
We thus conclude that with $O(\poly(n))$ time and $O(n\log n)$ space preprocessing, we can answer each query in $O(\sqrt{n})$ time. In the following, we reduce the preprocessing time to $O(n^{1+\delta})$ for any $\delta>0$, while the space and the query time do not change.

\subsubsection{Reducing the preprocessing time}
\label{sec:reducepre}

We build a partition tree $T$ by Lemma~\ref{lem:simpartitionalgo} recursively, until we obtain a partition of $P$ into subsets of size $s=O(n^{\delta'})$ for a suitable constant $\delta'$, which form the leaves of $T$.
Each inner node $v$ of $T$ corresponds to a subset $P(v)$ of $P$ and a simplicial
partition $\Pi(v)$ of $P(v)$, which form the children of $v$. At each child $u$ of
$v$, we store the cell $\triangle(u)$ of $\Pi(v)$ containing $P(u)$ and also store
the convex hull of $P(u)$.
Note that if $u$ is the root, then $\triangle(u)$ is the entire plane and $P(u)=P$.
We construct the partition $\Pi(v)$ using Lemma~\ref{lem:simpartitionalgo} with
parameter $s=|P(v)|^{1/2-\epsilon}$ for a small constant $\epsilon$. As such,
the height of $T$ is $O(1)$.
For each leaf $v$ of $T$, we build the data structure $\scrD_v$ discussed above
on $P(v)$, which takes time polynomial in $|P(v)|$. We make $\delta'$ small enough
so that the total time we spend on processing the leaves of $T$ is
$O(n^{1+\delta})$ for a given $\delta>0$. This finishes the preprocessing, which
takes $O(n^{1+\delta})$ time and $O(n\log n)$ space.

Given a query line $\ell$, our goal is to compute the point of $P$ closest to
$\ell$. We only discuss how to find the closest point of $\ell$ among all points
of $P$ below $\ell$ since the other case is similar. Starting from the root of
$T$, consider a node $v$. We assume that $\ell$ crosses $\triangle(v)$, which is
true initially when $v$ is the root. For each child $u$ of $v$, we do the
following. If $\ell$ crosses $\triangle(u)$, then we proceed on $u$ recursively.
Otherwise, if $\triangle(u)$ is below $\ell$, we do binary search on the convex
hull at $v$ to find in $O(\log n)$ time the closest point to $\ell$ among the
vertices of the convex hull and keep the point as a candidate. Eventually we
will reach a set $V$ of leaves $v$ whose cells $\triangle(v)$ are crossed by
$\ell$. Since the height of $T$ is $O(1)$, the size of $|V|$ is $O(\sqrt{r})$,
where $r=\sqrt{n/s}$ is the number of leaves of $T$. Also, because the height of
$T$ is $O(1)$, the number of nodes of $T$ visited by the above algorithm is
$O(n^{1/2-\epsilon'})$ for a constant $\epsilon'<\epsilon$. As binary search on
each convex hull takes $O(\log n)$ time, the total time the algorithm spends on
searching the tree $T$ is $O(n^{1/2-\epsilon'}\log n)$, which is bounded by
$O(n^{1/2})$. This also implies that the number of candidate closest points that
have been computed is $O(n^{1/2})$.

Next, for each leaf node $v\in V$, we use the data structure $\scrD_v$ to find
in $O(\sqrt{|P(v)|})$ time the closest point of $\ell$ among all points of $P(v)$
below $\ell$ and keep the point as a candidate point. As $|P(v)|=O(s)$, the total time for searching all leaf nodes of $V$ is $O(\sqrt{r}\cdot \sqrt{s})$, which is $O(\sqrt{n})$.
The above finds $O(\sqrt{n})$ candidate closest points in total.
Finally, among all candidate points we return the one closest to $\ell$.
The total query time is thus bounded by $O(\sqrt{n})$.
The following lemma summarizes the result.
\begin{lemma}\label{lem:linem2}
Given a set $P$ of $n$ points in the plane, we can build a data structure of
$O(n\log n)$ space in $O(n^{1+\delta})$ time for any $\delta>0$, such that for any query
line its closest point in $P$ can be computed in
$O(\sqrt{n})$ time.
\end{lemma}

\subsection{Trade-offs}
\label{sec:tradeoff}

Let $P^*$ denote the set of all lines dual to the points of $P$. Consider a query line
$\ell$. We assume that $\ell$ is not vertical. Let $\ell^*$ denote its dual
point.

To find the closest point of $\ell$ in $P$, it suffices to find
its closest point among all points of $P$ above $\ell$ and
its closest point among all points of $P$ below $\ell$. In the following, we
only compute its closest point among all points of $P$ above $\ell$.
In the dual plane, this is to find the first line
hit by the vertically downward ray $\rho(\ell^*)$ from $\ell^*$ among all dual
lines of $P^*$.

We compute a $(1/r)$-cutting $\Xi$ for $P^*$. This time instead of having each
cell of $\Xi$ as a triangle, we make each cell of $\Xi$ a trapezoid that is
bounded by two vertical edges, an upper edge, and a lower edge. This can be done
by slightly changing Chezelle's algorithm~\cite{ref:ChazelleCu93}, i.e., instead
of triangulating each cell of a line arrangement we decompose it into
trapezoids (i.e., draw a segment from each vertex of the cell until the cell boundary). Computing such a cutting $\Xi$ can be done in $O(nr)$ time~\cite{ref:ChazelleCu93}. A property of the cutting produced by Chezelle's algorithm~\cite{ref:ChazelleCu93} is that the upper/lower edge of each trapezoid
must lie on an line of $P^*$ unless it is unbounded. For each cell $\sigma$ of
$\Xi$, let $P^*(\sigma)$ denote the lines of $P^*$ crossing $\sigma$.
Hence, $|P^*(\sigma)|\leq n/r$. For each cell $\sigma$ of $\Xi$, we build a data structure of complexity $O(T(|P^*(\sigma)|),S(|P^*(\sigma)|),Q(|P^*(\sigma)|))$ for $P^*(\sigma)$, where $T(\cdot)$, $S(\cdot)$, and $Q(\cdot)$ are the preprocessing time, space, and query time, respectively; we refer to it as the {\em secondary data structure}. As $|P^*(\sigma)|\leq n/r$ and $\Xi$ has $O(r^2)$ cells, the total preprocessing time is $O(nr+r^2\cdot T(n/r))$ and the total space is $O(nr+r^2\cdot S(n/r))$.

For each query line $\ell$, we first find in $O(\log r)$ time the cell $\sigma$ of $\Xi$ that
contains the dual point $\ell^*$. Then, we use the secondary data structure for $P^*(\sigma)$ to find the first line $\ell'$ of $P^*(\sigma)$ hit by the ray $\rho(\ell^*)$. If $\ell'$ exists, then $\ell'$ is the solution; otherwise the line containing the lower side of $\sigma$ is the solution. The total query time is thus $O(\log r+Q(n/r))$.

As such, we obtain a data structure of complexity $O(nr+r^2\cdot T(n/r),nr+r^2\cdot S(n/r),\log r+Q(n/r))$, for any $1\leq r\leq n$. Using the results of Lemma~\ref{lem:linechan} and Lemma~\ref{lem:linem2} to build the secondary data structure for $P^*(\sigma)$, respectively, we can obtain the following trade-offs.

\begin{theorem}
\begin{enumerate}
  \item
      Given a set $P$ of $n$ points in the plane, we can build a data structure of
$O(nr\log (n/r))$ space in $O(nr\log (n/r))$ expected time, such that for any query
line its closest point in $P$ can be computed in $O(\sqrt{n/r})$ time with high probability, for any $1\leq r\leq n/\log^2 n$.
  \item
   Given a set $P$ of $n$ points in the plane, we can build a data structure of
$O(nr\log (n/r))$ space in $O(nr(n/r)^{\delta})$ time, such that for any query
line its closest point in $P$ can be computed in $O(\sqrt{n/r})$ time, for any $\delta>0$ and any $1\leq r\leq n/\log^2 n$.
\end{enumerate}
\end{theorem}

In particular, for the large space case, i.e., $r=n/\log^2n$, we can obtain a randomized data structure of complexity $O(n^2\log\log n/\log^2 n, n^2\log\log n/\log^2 n, \log n)$, and a slower deterministic data structure of complexity $O(n^2\log^{\delta} n/\log^2 n, n^2\log\log n/\log^2 n, \log n)$.

\section{The segment query problem}
\label{sec:segmentquery}

In this section, we discuss the query problem of the segment case. The main idea is essentially the same as the line case with one change:
whenever we compute the convex hull for a canonical subset of $P$ (e.g., the
subset $P(v)$ for a node $v$ in a partition tree) for outside-hull
line queries, we instead build the BS data structure~\cite{ref:BespamyatnikhQu00}
for outside-hull segment queries. Because the fractional cascading does not help
anymore, the query time in general has an additional logarithmic factor, with
the exception that when using Chan's partition tree~\cite{ref:ChanOp12} we still manage
to bound the query time by $O(\sqrt{n})$ due to some nice properties of the
partition tree.

\subparagraph{The randomized result.}
For our randomized result using Chan's partition
tree~\cite{ref:ChanOp12} (by modifying the one in Section~\ref{sec:chan}), for each node $v$ of the partition tree $T$, we construct the BS
data structure for $P(v)$. The total space is still $O(n\log n)$.
For the preprocessing time, constructing the BS data structure can be done in
linear time if we know the Voronoi diagram of $P(v)$. For this, as discussed in
Section~\ref{sec:review}, we can process
all nodes of $T$ in a bottom-up manner and using the linear-time Voronoi diagram
merge algorithm of Kirkpatrick~\cite{ref:KirkpatrickEf79}. As such, constructing
the BS data structures for all nodes of $T$ can be done in $O(n\log n)$ time in
total.
Therefore, the total preprocessing time of the data structure is still $O(n\log n)$ expected time.

The query algorithm follows the same scheme as before but instead use the BS
algorithm to answer outside-hull segment queries.
The total
query time becomes $O(\sqrt{n}\log n)$ with high probability. In fact, due to certain properties of Chan's
partition tree, the time is bounded by $O(\sqrt{n})$, as shown in the following
lemma (similar idea was used elsewhere, e.g., \cite{ref:ChanSi23}).

\begin{lemma}
The query time is bounded by $O(\sqrt{n})$ with high probability.
\end{lemma}
\begin{proof}
As discussed above, the extra logarithmic factor in the query time is due to the
outside-hull segment query, which takes $O(\log |P(v)|)$ time for each
node $v$ of $T$. We show that the sum of $\log |P(v)|$ for all nodes $v$ of $T$
visited by the query algorithm is $O(\sqrt{n})$ with high probability.

Indeed, $|P(v)|$ for each node $v$ at depth $i$ of $T$ (the depth of the root is
$0$) is bounded by $n/\rho^i$
for some constant $\rho>1$ and the height of $T$ is $k=\lceil\log_{\rho} n\rceil$~\cite{ref:ChanOp12}. Further, any line
$\ell$ crosses at most $O(\sqrt{\rho^i}+\log^{O(1)}n)$ cells at depth $i$ of $T$
with high probability~\cite{ref:ChanOp12}.
Hence, the total query time is bounded by
$\sum_{i=1}^k((\sqrt{\rho^i}+\log^{O(1)}n)\cdot \log(n/\rho^i))$ with high probability.
Next we show that $\sum_{i=1}^k((\sqrt{\rho^i}+\log^{O(1)}n)\cdot \log(n/\rho^i))=O(\sqrt{n})$. Indeed,
first notice that
$\sum_{i=1}^k\log^{O(1)}n\cdot \log(n/\rho^i)=O(\log^{O(1)}n)$, and thus is bounded by
$O(\sqrt{n})$.
For $\sum_{i=1}^k\sqrt{\rho^i}\cdot \log(n/\rho^i)$, since $\rho>1$, we have the following (for simplicity, we assume that $n=\rho^k$)
\begin{align*}
\sum_{i=1}^k\sqrt{\rho^i}\cdot \log(n/\rho^i) & =
\sum_{i=0}^{k-1}\sqrt{n/\rho^i}\cdot \log \rho^i =
\sqrt{n}\log \rho\cdot \sum_{i=0}^{k-1}\frac{i}{\sqrt{\rho^i}}=O(\sqrt{n}).
\end{align*}
This proves the lemma.
\end{proof}

As such, we obtain the following result.
\begin{lemma}\label{lem:segmentchan}
Given a set $P$ of $n$ segments in the plane, we can build a data structure of $O(n\log n)$ space in $O(n\log n)$ expected time such that for any query segment its closest point in $P$ can be computed in $O(\sqrt{n})$ time with high probability.
\end{lemma}

\subparagraph{The first deterministic result.}
For our first deterministic result (by modifying the one in Section~\ref{sec:linem1}), we again construct the BS data structure for $P(v)$ for each node $v$ of the partition tree $T$. The total space is still $O(n\log\log n)$. The query time follows the same scheme as before with changes as above, so is bounded by $O(\sqrt{n}\log^{O(1)}n)$.
For the preprocessing time, constructing the BS data structure for all nodes of $T$ can be done in $O(n\log n)$ time. Indeed, for each node $v$ of $T$, it takes $O(|P(v)|\log |P(v)|)$ time to construct the data structure. Because the subsets $P(v)$ for all nodes $v$ in the same level of $T$ are pairwise disjoint and the size of $P(v)$ is geometrically decreasing following the depth of $T$, the total time for constructing the data structures is bounded by $O(n\log n)$.
As such, we obtain the following result.
\begin{lemma}\label{lem:segmentm1}
Given a set $P$ of $n$ segments in the plane, we can build a data structure of $O(n\log\log n)$ space in $O(n\log n)$ time such that for any query segment its closest point in $P$ can be computed in $O(\sqrt{n}\log^{O(1)}n)$ time.
\end{lemma}

\subparagraph{The second deterministic result.}
For our second deterministic result (by modifying the one in Section~\ref{sec:linem2}), for each node $v$ of each tree $T_i$, we again construct the BS data structure for $P(v)$. We do the same for the data structure in Lemma~\ref{lem:firststep} as well as the tree $T$ in Section~\ref{sec:reducepre}. The total space is still $O(n\log n)$.
The preprocessing time for constructing the BS data structures for all nodes of $T$ is $O(n\log n)$ since $T$ has only $O(1)$ levels. Therefore, the total preprocessing time is still $O(n^{1+\delta})$. The query algorithm follows the same scheme as before with changes as above, so the query time becomes $O(\sqrt{n}\log n)$.

\begin{lemma}\label{lem:segmentm2}
Given a set $P$ of $n$ points in the plane, we can build a data structure of
$O(n\log n)$ space in $O(n^{1+\delta})$ time for any $\delta>0$, such that for any query
segment its closest point in $P$ can be computed in
$O(\sqrt{n}\log n)$ time.
\end{lemma}


\subparagraph{Trade-offs.}
For trade-offs, we now cannot use the previous approach for the line
case anymore because the dual problem in the segment case is not the vertical
ray-shootings among lines. Instead, we use an approach from
Bespamyatnikh~\cite{ref:BespamyatnikhCo03} (which is the algorithm for recurrence
\eqref{equ:20}).

Let $P^*$ denote the set of dual lines of the points of $P$. We compute a
$(1/r)$-hierarchical cutting $\Xi_0,\Xi_1,\ldots,\Xi_k$ for $P^*$. For each $i$,
$0\leq i<k$, for each cell $\sigma\in \Xi_i$, let $P^*(\sigma)$ denote the
subset of the lines of $P^*$ crossing $\sigma$; for each child $\sigma'$
of $\sigma$, let $P^*_1(\sigma',\sigma)$ (resp.,
$P^*_2(\sigma',\sigma)$) denote the subset of the lines $P^*(\sigma)$
above (resp., below) $\sigma'$. We construct a BS data structure on the dual
points of the lines of $P^*_1(\sigma',\sigma)$ (resp., $P^*_2(\sigma',\sigma)$).
For each cell $\sigma$ of the last cutting $\Xi_k$, we build a data structure
of complexity $O(T(|P^*(\sigma)|),S(|P^*(\sigma)|),Q(|P^*(\sigma)|))$ for the points dual to
the lines of $P^*(\sigma)$ and we refer to it as the {\em secondary data structure}.

The sum of $|P^*(\sigma)|$ for all cells $\sigma$ in $\Xi_i$
for all $i=0,1,\ldots,k$ is $O(nr)$. Since each cell has $O(1)$ children, the
total space of the cutting as well as the BS data structures for all cells
is $O(nr)$. The total space of the secondary data structures is $O(r^2\cdot
S(n/r))$. Hence, the overall preprocessing space is $O(nr+r^2\cdot S(n/r))$.

For the preprocessing time, constructing the cutting takes $O(nr)$ time.
Constructing the BS data structure on the dual
points of the lines of $P^*_1(\sigma',\sigma)$ (resp., $P^*_2(\sigma',\sigma)$) takes $O(|P^*(\sigma)|\log |P^*(\sigma)|)$ time for a child $\sigma'$ of a cell $\sigma\in \Xi_i$, $0\leq i\leq k$, and $|P^*(\sigma)|\leq n/\rho^i$, where $\rho$ is the constant associated with the hierarchical cutting as discussed in Section~\ref{sec:pre}. Also, each $\Xi_i$ has $O(\rho^{2i})$ cells and $k=\lceil\log_{\rho}r\rceil$.
Hence, the total time for
constructing the BS data structures for all cells of the cutting is on the order of (for simplicity we assume $r=\rho^{k}$)
\begin{align*}
\sum_{i=0}^{k}\rho^{2i}\cdot n/\rho^i\cdot \log(n/\rho^i) &=\sum_{i=0}^{k} n\cdot \rho^i\cdot \log(n/\rho^i)=n\cdot \sum_{i=0}^{k} r/\rho^i\cdot \log(n/r\cdot \rho^i) \\
& =nr \cdot \sum_{i=0}^{k} \frac{\log(n/r) + i\cdot \log \rho}{\rho^i} = O(nr\log (n/r)).
\end{align*}
As such, the total preprocessing time is $O(nr\log (n/r)+r^2\cdot T(n/r))$.

Given a query segment $s$, let $s^*$ be the point dual to the supporting line of
$s$. We first find the cell $\sigma_i$ of $\Xi_i$ containing $s^*$ for all
$i=0,1,\ldots,k$, which takes $O(\log r)$ time.
For each $1\leq i\leq k$, we find the closest point to $s$ among the dual points
of $P^*_1(\sigma',\sigma)$ (resp., $P^*_2(\sigma',\sigma)$) using
the BS data structure in $O(\log n)$ time and keep them as
candidate closest points. In addition, for $\sigma_k$, we find the closest
point to $s$ among the dual points of $P^*(\sigma)$ using the secondary data
structure and keep it as a candidate closest point.
Finally, among
all $O(\log r)$ candidate closest points, we return the one closest to $s$ as
the answer. Hence, the total
query time is $O(\log r\log n+Q(n/r))$.


As such, we obtain a data structure of complexity $O(nr\log (n/r)+r^2\cdot T(n/r),nr+r^2\cdot S(n/r),\log r\log n+Q(n/r))$, for any $1\leq r\leq n$. If we use the results of Lemmas~\ref{lem:segmentchan} and \ref{lem:segmentm2} as secondary data structures, respectively, we can obtain the following trade-offs.

\begin{theorem}
\begin{enumerate}
  \item
  Given a set $P$ of $n$ points in the plane, we can build a data structure of
$O(nr\log (n/r)))$ space in $O(nr\log (n/r))$ expected time, such that for any query
segment its closest point in $P$ can be computed in $O(\sqrt{n/r})$ time with high probability, for any $1\leq r\leq n/\log^4 n$.
  \item
      Given a set $P$ of $n$ points in the plane, we can build a data structure of
$O(nr\log (n/r))$ space in $O(nr(n/r)^{\delta})$ time, such that for any query
segment its closest point in $P$ can be computed in $O(\sqrt{n/r}\cdot\log (n/r))$ time, for any $\delta>0$ and any $1\leq r\leq n\log^2\log n/\log^4 n$.
\end{enumerate}
\end{theorem}

In particular, for the large space case, we can obtain a randomized data structure of complexity $O(n^2\log\log n/\log^4 n, n^2\log\log n/\log^4 n, \log^2 n)$ when $r=n/\log^4 n$, and a slower deterministic data structure of complexity $O(n^2/\log^{4-\delta} n, n^2\log^3\log n/\log^4 n, \log^2 n)$ when $r=n\log^2\log n/\log^4n$.

\subparagraph{A randomized $\boldsymbol{O(n^{4/3},n^{4/3},n^{1/3})}$ solution.} We finally remark on another randomized solution of complexity $O(n^{4/3},n^{4/3},n^{1/3})$, by using Chan's randomized techniques~\cite{ref:ChanGe99} and Chan and Zheng's recent randomized result on triangle range counting~\cite{ref:ChanHo22}.\footnote{The idea was suggested by an anonymous reviewer.}
First of all, we consider the following {\em decision version} of the query problem: Given a query segment $s$ and a value $\delta$, determine whether $P$ has a point whose distance from $s$ is no more than $\delta$. Let $N_{\delta}(s)$ denote the region of the plane consisting of all points whose distances from $s$ is no more than $\delta$. The decision problem is thus to decide whether $N_{\delta}(s)\cap P=\emptyset$. Observe that $N_{\delta}(s)$ is the union of two disks and a rectangle (see Fig.~\ref{fig:queryseg}). More specifically, the two disks have the same radius equal to $\delta$ and have the two endpoints of $s$ as their centers, respectively; the rectangle is symmetric with respect to $s$. Let $D_{\delta}(s)$ denote the union of the two disks and $R_{\delta}(s)$ the rectangle. It suffices to decide whether $D_{\delta}(s)\cap P= \emptyset$ and $R_{\delta}(s)\cap P= \emptyset$. The former can be determined in $O(\log n)$ time after constructing a Voronoi diagram on $P$ in $O(n\log n)$ time and $O(n)$ space. The latter can be determined using the randomized result of Chan and Zheng~\cite{ref:ChanHo22} on triangle range counting. Specifically, we can build a data structure of $O(n^{4/3})$ space in $O(n^{4/3})$ expected time on $P$ so that the number of points of $P$ inside a query triangle can be computed in $O(n^{1/3})$ expected time~\cite{ref:ChanHo22}. As such, with two triangle range counting queries, we can determine whether $R_{\delta}(s)\cap P= \emptyset$ in $O(n^{1/3})$ expected time. In summary, with $O(n^{4/3})$ space and $O(n^{4/3})$ expected preprocessing time by a randomized algorithm, we can solve the decision problem for any query segment $s$ in $O(n^{1/3})$ expected time.

\begin{figure}[t]
\begin{minipage}[t]{\textwidth}
\begin{center}
\includegraphics[height=1.5in]{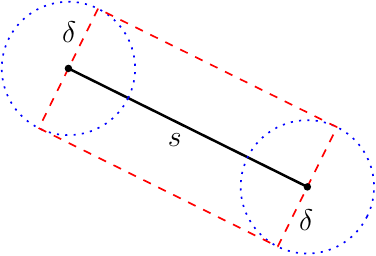}
\caption{\footnotesize $D_{\delta}(s)$ is the union of the two (blue) dotted disks (whose radii are both $\delta$) and $R_{\delta}(s)$ is the (red) dashed rectangle. $N_{\delta}(s)$ is the union of $D_{\delta}(s)$ and $R_{\delta}(s)$.}
\label{fig:queryseg}
\end{center}
\end{minipage}
\vspace{-0.15in}
\end{figure}

Next, to solve the original query problem (i.e., given a query segment $s$, find its closest point in $P$), we can simply apply Chan's randomized techniques~\cite{ref:ChanGe99} to reduce the problem to the above decision problem. Indeed, the problem reduction argument is quite similar to the problem reduction of the ray shooting problem (see Section 3.2 in~\cite{ref:ChanGe99}) and we thus omit the details. We therefore conclude with the following.

\begin{theorem}\label{theo:100}
Given a set $P$ of $n$ points in the plane, we can build a data structure of
$O(n^{4/3})$ space in $O(n^{4/3})$ expected time by a randomized algorithm, such that for any query
segment its closest point in $P$ can be computed in $O(n^{1/3})$ expected time.
\end{theorem}

As discussed in Section~\ref{sec:intro}, using the above result, we can solve the offline segment-closest-point problem for $n$ points and $n$ segments in $O(n^{4/3})$ expected time. Since the segment-closest-point problem has an $\Omega(n^{4/3})$ time lower bound in the model of Erickson~\cite{ref:EricksonNe96}, the result in Theorem~\ref{theo:100} is optimal (at least in term of proprocessing/query time) under the restricted model of Erickson~\cite{ref:EricksonNe96}.

\section*{Acknowledgment}
The author would like to thank an anonymous reviewer for suggesting the idea of using Chan's randomized techniques~\cite{ref:ChanGe99} and Chan and Zheng's recent randomized result on triangle range counting~\cite{ref:ChanHo22} to tackle the problem.




\begin{thebibliography}{10}

\bibitem{ref:AgarwalPa902}
Pankaj~K. Agarwal.
\newblock Partitioning arrangements of lines {II}: Applications.
\newblock {\em Discrete and Computational Geometry}, 5:533--573, 1990.

\bibitem{ref:AgarwalAp93}
Pankaj~K. Agarwal and Micha Sharir.
\newblock Applications of a new space-partitioning technique.
\newblock {\em Discrete and Computational Geometry}, 9:11--38, 1993.

\bibitem{ref:AronovSu21}
Boris Aronov, Mark de~Berg, Jean Cardinal, Esther Ezra, John Iacono, and Micha
  Sharir.
\newblock Subquadratic algorithms for some {3Sum-Hard} geometric problems in
  the algebraic decision tree model.
\newblock In {\em Proceedings of the 32nd International Symposium on Algorithms
  and Computation (ISAAC)}, pages 3:1--3:15, 2021.

\bibitem{ref:Bar-YehudaVa94}
Reuven Bar-Yehuda and Sergio Fogel.
\newblock Variations on ray shootings.
\newblock {\em Algorithmica}, 11:133--145, 1994.

\bibitem{ref:BenderTh00}
Michael~A. Bender and Martin Farach-Colton.
\newblock The {LCA} problem revisited.
\newblock In {\em Proceedings of the 4th Latin American Symposium on
  Theoretical Informatics}, pages 88--94, 2000.

\bibitem{ref:BespamyatnikhCo03}
Sergei Bespamyatnikh.
\newblock Computing closest points for segments.
\newblock {\em International Journal of Computational Geometry and
  Application}, 13:419--438, 2003.

\bibitem{ref:BespamyatnikhQu00}
Sergei Bespamyatnikh and Jack Snoeyink.
\newblock Queries with segments in {Voronoi} diagrams.
\newblock {\em Computational Geometry: Theory and Applications}, 16:23--33,
  2000.

\bibitem{ref:ChanGe99}
Timothy~M. Chan.
\newblock Geometric applications of a randomized optimization technique.
\newblock {\em Discrete and Computational Geometry}, 22:547--567, 1999.

\bibitem{ref:ChanOp12}
Timothy~M. Chan.
\newblock Optimal partition trees.
\newblock {\em Discrete and Computational Geometry}, 47:661--690, 2012.

\bibitem{ref:ChanAn24}
Timothy~M. Chan, Pingan Cheng, and Da~Wei Zheng.
\newblock An optimal algorithm for higher-order voronoi diagrams in the plane:
  The usefulness of nondeterminism.
\newblock In {\em Proceedings of the 35th Annual ACM-SIAM Symposium on Discrete
  Algorithms (SODA)}, 2024.

\bibitem{ref:ChanHo22}
Timothy~M. Chan and Da~Wei Zheng.
\newblock Hopcroft’s problem, log-star shaving, {2D} fractional cascading,
  and decision trees.
\newblock In {\em Proceedings of the 33rd Annual ACM-SIAM Symposium on Discrete
  Algorithms (SODA)}, pages 190--210, 2022.

\bibitem{ref:ChanSi23}
Timothy~M. Chan and Da~Wei Zheng.
\newblock Simplex range searching revisited: How to shave logs in multi-level
  data structures.
\newblock In {\em Proceedings of the 34th Annual ACM-SIAM Symposium on Discrete
  Algorithms (SODA)}, pages 1493--1511, 2023.

\bibitem{ref:ChazelleAn88}
Bernard Chazelle.
\newblock An algorithm for segment-dragging and its implementation.
\newblock {\em Algorithmica}, 3:205--221, 1988.

\bibitem{ref:ChazelleCu93}
Bernard Chazelle.
\newblock Cutting hyperplanes for divide-and-conquer.
\newblock {\em Discrete and Computational Geometry}, 9:145--158, 1993.

\bibitem{ref:ChazelleFr86}
Bernard Chazelle and Leonidas~J. Guibas.
\newblock Fractional cascading: {I. A} data structuring technique.
\newblock {\em Algorithmica}, 1:133--162, 1986.

\bibitem{ref:ChazelleTh85}
Bernard Chazelle, Leonidas~J. Guibas, and D.T. Lee.
\newblock The power of geometric duality.
\newblock {\em BIT}, 25:76--90, 1985.

\bibitem{ref:ChengAl92}
Siu~Wing Cheng and Ravi Janardan.
\newblock Algorithms for ray-shooting and intersection searching.
\newblock {\em Journal of Algorithms}, 13:670--692, 1992.

\bibitem{ref:ColeGe83}
Richard Cole and Chee{-}Keng Yap.
\newblock Geometric retrieval problems.
\newblock In {\em Proceedings of the 24th Annual Symposium on Foundations of
  Computer Science (FOCS)}, pages 112--121, 1983.

\bibitem{ref:DaescuFa06}
Ovidiu Daescu, Ningfang Mi, Chan{-}Su Shin, and Alexander Wolff.
\newblock Farthest-point queries with geometric and combinatorial constraints.
\newblock {\em Computational Geometry: Theory and Applications}, 33:174--185,
  2006.

\bibitem{ref:DriscollMa89}
James~R. Driscoll, Neil Sarnak, Daniel~D. Sleator, and Robert~E. Tarjan.
\newblock Making data structures persistent.
\newblock {\em Journal of Computer and System Sciences}, 38:86--124, 1989.

\bibitem{ref:EdelsbrunnerTh90}
Herbert Edelsbrunner, Leonidas~J. Guibas, and Micha Sharir.
\newblock The complexity and construction of many faces in arrangement of lines
  and of segments.
\newblock {\em Discrete and Computational Geometry}, 5:161--196, 1990.

\bibitem{ref:EdelsbrunnerOp86}
Herbert Edelsbrunner, Leonidas~J. Guibas, and J.~Stolfi.
\newblock Optimal point location in a monotone subdivision.
\newblock {\em SIAM Journal on Computing}, 15(2):317--340, 1986.

\bibitem{ref:EricksonNe96}
Jeff Erickson.
\newblock New lower bounds for {Hopcroft's} problem.
\newblock {\em Discrete and Computational Geometry}, 16:389--418, 1996.

\bibitem{ref:FredmanHo76}
Michael~L. Fredman.
\newblock How good is the information theory bound in sorting?
\newblock {\em Theoretical Computer Science}, 1:355--361, 1976.

\bibitem{ref:GoodmanMu83}
Jacob~E. Goodman and Richard Pollack.
\newblock Multidimensional sorting.
\newblock {\em SIAM Journal on Computing}, 12:484--507, 1983.

\bibitem{ref:GoodmanUp86}
Jacob~E. Goodman and Richard Pollack.
\newblock Upper bounds for configurations and polytopes in {$R^d$}.
\newblock {\em Discrete and Computational Geometry}, pages 219--227, 1986.

\bibitem{ref:GoswamiTr04}
Partha~P. Goswami, Sandip Das, and Subhas~C. Nandy.
\newblock Triangular range counting query in {2D} and its application in
  finding {$k$} nearest neighbors of a line segment.
\newblock {\em Computational Geometry: Theory and Applications}, 29:163--175,
  2004.

\bibitem{ref:HarelFa84}
Dov Harel and Robert~E. Tarjan.
\newblock Fast algorithms for finding nearest common ancestors.
\newblock {\em SIAM Journal on Computing}, 13:338--355, 1984.

\bibitem{ref:KirkpatrickEf79}
David~G. Kirkpatrick.
\newblock Efficient computation of continuous skeletons.
\newblock In {\em Proceedings of the 20th Annual Symposium on Foundations of
  Computer Science (FOCS)}, pages 18--27, 1979.

\bibitem{ref:KirkpatrickOp83}
David~G. Kirkpatrick.
\newblock Optimal search in planar subdivisions.
\newblock {\em SIAM Journal on Computing}, 12(1):28--35, 1983.

\bibitem{ref:LeeTh85}
D.~T. Lee and Y.~T. Ching.
\newblock The power of geometric duality revisited.
\newblock {\em Information Processing Letters}, 21:117--122, 1985.

\bibitem{ref:MatousekEf92}
Ji\u{r}\'{i} Matou\v{s}ek.
\newblock Efficient partition trees.
\newblock {\em Discrete and Computational Geometry}, 8:315--334, 1992.

\bibitem{ref:MatousekRa93}
Ji\u{r}\'{i} Matou\v{s}ek.
\newblock Range searching with efficient hierarchical cuttings.
\newblock {\em Discrete and Computational Geometry}, 10:157--182, 1993.

\bibitem{ref:MitraEf98}
Pinaki Mitra and Bidyut~B. Chaudhuri.
\newblock Efficiently computing the closest point to a query line.
\newblock {\em Pattern Recognition Letters}, 19:1027--1035, 1998.

\bibitem{ref:MukhopadhyayUs03}
Asish Mukhopadhyay.
\newblock Using simplicial paritions to determine a closest point to a query
  line.
\newblock {\em Pattern Recognition Letters}, 24:1915--1920, 2003.

\bibitem{ref:OvermarsMa81}
Mark~H. Overmars and Jan van Leeuwen.
\newblock Maintenance of configurations in the plane.
\newblock {\em Journal of Computer and System Sciences}, 23:166--204, 1981.

\bibitem{ref:SnarnakPl86}
Neil Sarnak and Robert~E. Tarjan.
\newblock Planar point location using persistent search trees.
\newblock {\em Communications of the ACM}, 29:669--679, 1986.

\bibitem{ref:WangAl20}
Haitao Wang.
\newblock Algorithms for subpath convex hull queries and ray-shooting among
  segments.
\newblock In {\em Proceedings of the 36th International Symposium on
  Computational Geometry (SoCG)}, pages 69:1--69:14, 2020.

\end{thebibliography}

\end{document}